\documentclass[journal,onecolumn,10pt,twoside]{IEEEtranTCOM}
\newtheorem{definition}{{Definition}}
\newtheorem{theorem}{{Theorem}}
\newtheorem{lemma}{{Lemma}}

\usepackage{cite}
\usepackage{amsmath,amssymb,amsfonts}
\usepackage{algorithmic}
\usepackage{graphicx}
\usepackage{textcomp}
\usepackage{xcolor}
\usepackage{bbm}

\normalsize

\ifCLASSINFOpdf
\else
\fi

\begin{document}
\title{Feedback Does Not Increase the Capacity of  Approximately Memoryless Surjective POST Channels}

\author{Xiaojing~Zhang, Jun~Chen, Guanghui~Wang % stops a space
%\thanks{Xiaojing Zhang and Guanghui Wang are with the School of Mathematics, Shandong University, Jinan, Shandong 250100, China (e-mail: xiaojingzhang@mail.sdu.edu.cn, ghwang@sdu.edu.cn).}
%\thanks{Jun Chen is with the Department of Electrical and Computer Engineering, McMaster University, Hamilton, ON L8S 4K1, Canada (e-mail: chenjun@mcmaster.ca).}
%\thanks{The work of Xiaojing Zhang and Guanghui Wang was supported in part by the National Key R\&D Program of China under Grant No. 2023YFA1009600.}
%\thanks{Guanghui Wang is with the School of Mathematics and the State Key Laboratory of Cryptography and Digital Economy Security, Shandong University, Jinan, Shandong 250000, China (e-mail: ghwang@sdu.edu.cn)}
% <-this % stops a space
}

\maketitle

\begin{abstract}
%\boldmath
We study a class of finite-state channels, known as POST channels, in which the previous channel output serves as the current state. A POST channel is deemed approximately memoryless when the state-dependent transition matrices are sufficiently close to one another. For this family of channels, under a surjectivity condition on the associated memoryless reference channel, we show that the feedback capacity coincides with the non-feedback capacity. Consequently, for almost all approximately memoryless POST channels whose input alphabet size is no smaller than the output alphabet size, feedback provides no capacity gain. This result extends Shannon’s classical theorem on discrete memoryless channels and demonstrates that the phenomenon holds well beyond the strictly memoryless case.  
\end{abstract}


% \fi
\IEEEpeerreviewmaketitle

\section{Introduction}\label{sec:introduction}

\IEEEPARstart{O}{ne}
 of the most celebrated results in information theory is Shannon’s theorem, which asserts that feedback does not increase the capacity of discrete memoryless channels \cite[Theorem 6]{Shannon56} (see also \cite[Theorem 8.12.1]{CT91}). This naturally raises the question of whether the converse holds: if feedback fails to increase a channel’s capacity, must the channel be memoryless? Equivalently, does every channel with memory experience a capacity gain when feedback is allowed? If true, this would mean that Shannon’s theorem identifies a property unique to memoryless channels.
 
It turns out that there exist channels with memory whose feedback capacity coincides with their non-feedback capacity. A well-known example is given in \cite{Alajaji95}, which shows that feedback does not increase the capacity of additive-noise channels, irrespective of the memory structure of the noise process. The reason is that the cyclic symmetry of the modulo-addition operation inherently favors i.i.d. uniform inputs, leaving no room for feedback to improve performance. Other examples of channels with memory for which feedback provides no capacity gain can also be constructed. However, such constructions are typically rather contrived: they either exhibit symmetry properties that eliminate the need for input adaptation, rendering feedback useless, or they possess special structural features that allow the encoder to infer the channel output without feedback. For these reasons, it remains widely believed that, for a generic channel with memory, the feedback capacity should exceed the non-feedback capacity. This intuition is natural, as feedback enables the encoder to obtain more information about the evolving channel state and, consequently, to adapt the transmission strategy to achieve a higher communication rate.

In this paper, we study a class of finite-state channels known as POST channels, in which the previous channel output serves as the current channel state. A POST channel is said to be approximately memoryless when the state-dependent transition matrices  are sufficiently close to one another. 
For this family of channels, under a surjectivity condition on the associated memoryless reference channel, we establish that the feedback capacity coincides with the non-feedback capacity. Consequently, for almost all approximately memoryless POST channels whose input alphabet size is no smaller than the output alphabet size, feedback does not increase capacity. Our work is motivated by the striking result in \cite[Section V]{PAW14}, which shows that feedback offers no capacity gain for the binary POST$(a,b)$ channel---a binary-input, binary-output POST channel whose crossover probabilities are $(1-a,1-b)$ in state $0$ and $(1-b,1-a)$ in state $1$. However, due to the strong symmetry of this example, it is difficult to determine whether the phenomenon is truly fundamental or merely coincidental---namely, whether a clever coding scheme could effectively transform the POST$(a,b)$ channel into a memoryless channel with crossover probabilities $(1-a,1-b)$. Our result confirms that the phenomenon is in fact much more general. Moreover, since approximately memoryless POST channels lack such symmetry, our result also rules out the possibility that the channel model is implicitly degraded to a memoryless one. From a broader perspective, our findings extend Shannon’s classical theorem, as a discrete memoryless channel can be regarded as a degenerate case of an approximately memoryless POST channel in which all state-dependent transition matrices coincide. Thus, rather than viewing memoryless channels as the only channels for which feedback fails to increase capacity, one should view them as an extreme instance of a much larger class exhibiting this property.

The remainder of the paper is organized as follows. Section \ref{sec:definition} reviews the necessary background on POST channels and presents the information-theoretic characterizations of their non-feedback and feedback capacities, along with a formal statement of the main result. Section \ref{sec:proof} contains the proof of the main result. A conceptual illustration of the core ideas in this work is presented in Section \ref{sec:discussion}.
Finally, Section \ref{sec:conclusion} offers concluding remarks and discusses possible directions for future research.

We adopt the following notational conventions. Random variables are denoted by uppercase letters (e.g., 
$X$), and their realizations by lowercase letters (e.g., 
$x$). For sequences, we write 
$X^j_i$
as shorthand for 
$(X_i,X_{i+1},\ldots,X_j)$, and use 
$X^j$
when 
$i=1$; $X^j_i$ is taken to be a null variable when $i>j$. Calligraphic letters denote sets (e.g., 
$\mathcal{S}$), with 
$|\mathcal{S}|$ representing the cardinality of 
$\mathcal{S}$. Matrices are written in uppercase blackboard bold (e.g., 
$\mathbb{Q}$), while lowercase blackboard bold symbols (e.g., 
$\mathbbm{q}$) denote column vectors. In particular, 
$\mathbb{I}$ denotes the identity matrix and 
$\mathbbm{1}$ the all-ones column vector, with dimensions clear from context. The infinity norm 
$\|\cdot\|_{\infty}$
refers to the vector max norm; when applied to a matrix, it denotes the maximum absolute entry. We use $\|\cdot\|_1$,
$\|\cdot\|_2$
and 
$\|\cdot\|_F$
for the Manhattan, spectral, and Frobenius norms, respectively, and 
$1\{\cdot\}$ for the indicator function. Unless otherwise specified, all logarithms are taken to base 
$e$.

\section{Non-Feedback Capacity vs. Feedback Capacity}\label{sec:definition}

A POST channel is a finite-state channel where the previous channel output serves as the current channel state. Formally, a POST channel is characterized by a transition kernel $Q(y|x,y')$, with $x\in\mathcal{X}$ and $y,y'\in\mathcal{Y}$, which specifies  conditional probability of the  output  $Y_t=y$ given the  input $X_t=x$ and the  state $Y_{t-1}=y'$ at time $t$ for  $t=1,2,\ldots$ Without loss of generality, we assume $\mathcal{X}=\{0,1,\ldots,|\mathcal{X}|-1\}$ and $\mathcal{Y}=\{0,1,\ldots,|\mathcal{Y}|-1\}$ with $|\mathcal{X}|, |\mathcal{Y}|\geq 2$.

To facilitate analysis, it is  useful to introduce  matrix representations\footnote{For notational convenience, matrix rows and columns are indexed starting from 
$0$, unless stated otherwise.} of the transition kernel $Q$. Specifically, for each $y'\in\mathcal{Y}$, let $\mathbb{Q}^{(c)}_{y'}$ denote the channel transition matrix corresponding to state $y'$, with its $(y,x)$-entry  $\mathbb{Q}^{(c)}_{y'}(y,x):=Q(y|x,y')$. Likewise, for each $x\in\mathcal{X}$, let  $\mathbb{Q}^{(s)}_{x}$ denote the state transition matrix corresponding to channel input $x$, with its $(y,y')$-entry   $\mathbb{Q}^{(s)}_{x}(y,y'):=Q(y|x,y')$. These matrices provide two complementary views of the POST channel dynamics---one indexed by the state and the other by the input.

%define POST channel, is specified by transition kernel $W$, which determines for all $t$

\subsection{Non-Feedback Capacity}

A length-$n$ non-feedback coding system for a POST channel $Q$ consists of an encoding function $f^{(n)}:\mathcal{M}^{(n)}\rightarrow\mathcal{X}^n$, which maps a message $M$, uniformly distributed over $\mathcal{M}^{(n)}$, to a channel input sequence $X^n$, and a decoding function $g^{(n)}:\mathcal{Y}^n\rightarrow\mathcal{M}^{(n)}$, which maps the channel output sequence $Y^n$ to a reconstructed message $\hat{M}$. Conditioned on the initial state $Y_0=y_0$, the  joint distribution induced by this non-feedback coding system factors as
\begin{align}
	P_{MX^nY^n\hat{M}|Y_0}(m,x^n,y^n,\hat{m}|y_0)=\frac{1}{|\mathcal{M}^{(n)}|}1\{x^n=f^{(n)}(m)\}\left(\prod_{t=1}^nQ(y_t|x_t,y_{t-1})\right)1\{\hat{m}=g^{(n)}(y^n)\}.\label{eq:non-feedback distribution}
\end{align}

\begin{definition}[Non-Feedback Capacity]
A rate $R$ is said to be worst-case achievable without feedback for a POST channel $Q$ if there exists a sequence of non-feedback coding systems $\{(f^{(n)},g^{(n)})\}_{n=1}^{\infty}$ satisfying
\begin{align}
	&\lim\limits_{n\rightarrow\infty}\frac{1}{n}\log|\mathcal{M}^{(n)}|= R,\label{eq:worst_rate}\\
	&\lim\limits_{n\rightarrow\infty}\max\limits_{y_0\in\mathcal{Y}}\mathrm{Pr}\{\hat{M}\neq M|Y_0=y_0\}=0.\label{eq:worst_error}
\end{align}
A best-case achievable rate without feedback for $Q$ is defined similarly, with  \eqref{eq:worst_error} replaced by
\begin{align}
	\lim\limits_{n\rightarrow\infty}\min\limits_{y_0\in\mathcal{Y}}\mathrm{Pr}\{\hat{M}\neq M|Y_0=y_0\}=0.\label{eq:best_error}
\end{align}
The maximum worst-case achievable rate and the maximum best-case achievable rate without feedback are denoted by $\underline{C}(Q)$ and $\overline{C}(Q)$, respectively.  When these two quantities coincide, their common value is denoted by $C(Q)$ and is referred to as the non-feedback capacity of the POST channel $Q$. 
\end{definition}

A POST channel $Q$ is said to be indecomposable if, for every $\epsilon>0$, there exists an integer $N$ such that for all $n\geq N$, the product of any $n$ matrices from $\{\mathbb{Q}^{(s)}_{0},\mathbb{Q}^{(s)}_{1},\ldots,\mathbb{Q}^{(s)}_{|\mathcal{X}|-1}\}$ has all entries within each row differing by at most $\epsilon$. Since the $(y,y')$-entry of $\mathbb{Q}^{(s)}_{x_n}\cdots\mathbb{Q}^{(s)}_{x_2}\mathbb{Q}^{(s)}_{x_1}$ represents the conditional probability of $Y_t=y$ given $X^n=x^n$ and $Y_0=y'$, the indecomposability condition ensures that the influence of the initial state on future states becomes negligible over time. Define
\begin{align}
	C_n(Q):=\max\limits_{P_{X^nY_0}}\frac{1}{n}I(X^n;Y^n|Y_0),\label{eq:C_n(Q)}
\end{align}
where the joint distribution underlying the conditional mutual information expression factors as
\begin{align}
	P_{X^nY^nY_0}(x^n,y^n,y_0)=P_{X^nY_0}(x^n,y_0)\left(\prod\limits_{t=1}^nQ(y_t|x_t,y_{t-1})\right).
\end{align} 
As shown in
\cite[Theorems 3 and 4]{BBT58}
and
\cite[Theorems 4.6.2, 4.6.4, and 5.9.2]{Gallager68}, for any indecomposable POST channel $Q$,
\begin{align}
	C(Q)=C^*(Q):=\lim\limits_{n\rightarrow\infty}C_n(Q).\label{eq:nonfeedback_capacity}
\end{align}
In general, computing 
$C(Q)$ is difficult because its characterization involves the limit of an 
$n$-letter expression.

\subsection{Feedback Capacity}

We now turn to the setting with feedback. A length-$n$ feedback coding system for a POST channel $Q$ consists of  an encoding chain $f^{(n)}_1,f^{(n)}_2,\ldots,f^{(n)}_n$, with $f^{(n)}_t: \mathcal{M}^{(n)}\times\mathcal{Y}^{t-1}\rightarrow\mathcal{X}$ generating the channel input symbol $X_t$ based on the  given message $M$ and the past output symbols $Y^{t-1}$ for $t=1,2,\ldots,n$, and a decoding function $g^{(n)}:\mathcal{Y}^n\rightarrow\mathcal{M}^{(n)}$, which maps the channel output sequence $Y^n$ to a reconstructed message $\hat{M}$. Conditioned on the initial state $Y_0=y_0$, the  joint distribution induced by this feedback coding system factors as
\begin{align}
	P_{MX^nY^n\hat{M}|Y_0}(m,x^n,y^n,\hat{m}|y_0)=\frac{1}{|\mathcal{M}^{(n)}|}\left(\prod\limits_{t=1}^n1\{x_t=f^{(n)}_t(m,y^{t-1})\}\right)\left(\prod_{t=1}^nQ(y_t|x_t,y_{t-1})\right)1\{\hat{m}=g^{(n)}(y^n)\}.\label{eq:feedback distribution}
\end{align}

\begin{definition}[Feedback Capacity]
A rate $R$ is said to be worst-case achievable with feedback for a POST channel $Q$ if there exists a sequence of feedback coding systems $\{(f^{(n)}_1,f^{(n)}_2,\ldots,f^{(n)}_n, g^{(n)})\}_{n=1}^{\infty}$ satisfying \eqref{eq:worst_rate} and \eqref{eq:worst_error}. A best-case achievable rate with feedback for $Q$ is defined similarly, with  \eqref{eq:worst_error} replaced by \eqref{eq:best_error}.
The maximum worst-case achievable rate  and the maximum best-case achievable rate with feedback are denoted by $\underline{C}_f(Q)$ and $\overline{C}_f(Q)$, respectively.  When these two quantities coincide, their common value is denoted by $C_f(Q)$ and is referred to as the feedback capacity of the POST channel $Q$. 
\end{definition}

A POST channel $Q$ is said to be connected if, for any $y,y'\in\mathcal{Y}$, the product of some matrices from $\{\mathbb{Q}^{(s)}_{0},\mathbb{Q}^{(s)}_{1},\ldots,\mathbb{Q}^{(s)}_{|\mathcal{X}|-1}\}$ has a positive $(y,y')$-entry. This connectivity condition guarantees that, for any initial state 
$y'$
and target state 
$y$, there exists a sequence of channel inputs under which the channel can reach 
$y$ from 
$y'$
with positive probability.
As shown in \cite[Theorem 1]{SSP22} (see also \cite[Theorems 1, 4, and 5]{CB05}), for any connected POST channel $Q$,
\begin{align}
	C_f(Q)=C^*_f(Q):=&\max\limits_{P_{XY'}} I(X;Y|Y')\label{eq:feedback_capacity}\\
	&\mbox{s.t.}\quad P_{Y}=P_{Y'},\label{eq:feedback_constraint}
\end{align}
where the joint distribution underlying the mutual information expression factors as
\begin{align}
	P_{XYY'}(x,y,y')=P_{XY'}(x,y')Q(y|x,y'),\label{eq:feedback_factorization}
\end{align}
and $P_Y$ denotes the marginal distribution of $Y$ induced by this joint distribution. Thanks to this single-letter characterization, $C_f(Q)$ is generally more amenable to evaluation that $C(Q)$. In fact, the maximization problem in \eqref{eq:feedback_capacity}--\eqref{eq:feedback_constraint} is a convex program, which can be readily solved using standard convex-optimization methods.

\subsection{Main Results}

A POST channel $Q$ reduces to a memoryless channel if $Q(y|x,y')$ does not depend on $y'$. We say that $Q$ is approximately memoryless if $\max_{y',y''\in\mathcal{Y}}\|\mathbb{Q}^{(c)}_{y'}-\mathbb{Q}^{(c)}_{y''}\|_{\infty}$ is small. For ease of presentation, it is often more convenient to express this requirement by measuring the proximity of $Q$ to a fixed memoryless reference  channel $W$. Specifically, we will refer to a POST channel $Q$ as $W$-centered $\delta$-approximately memoryless if
\begin{align} \max_{y'\in\mathcal{Y}}\|\mathbb{Q}^{(c)}_{y'}-\mathbb{W}\|_{\infty}\leq\delta,\label{eq:infinity_norm}
\end{align}
where $\mathbb{W}$ denotes the matrix representation of $W$, with its $(y,x)$-entry $\mathbb{W}(y,x):=W(y|x)$.

 %This convention will be adopted throughout the remainder of the paper.

The following result shows that a $\delta$-approximately memoryless POST channel is guaranteed to be indecomposable and connected
for $\delta$ chosen sufficiently small.

\begin{lemma}\label{lem:indecomposable_connected}
A $W$-centered $\delta$-approximately memoryless POST channel $Q$ is indecomposable whenever
\begin{align} \delta<\min_{x\in\mathcal{X}}\max_{y\in\mathcal{Y}}W(y|x)\label{eq:cond_indecomposable}
\end{align}	
 and  connected whenever
 \begin{align}
 \delta<\min_{y\in\mathcal{Y}}\max_{x\in\mathcal{X}}W(y|x).\label{eq:cond_connected}
 \end{align} 
\end{lemma} 
\begin{IEEEproof}
	See Appendix \ref{app:indecomposable_connected}.
	\end{IEEEproof}

%Input-rich (needs to be defined in the abstract and introduction)

A memoryless channel $W$  is said to be surjective  if 
\begin{enumerate}
\item (strict complementary slackness) 
\begin{align}
	&P^{(W)}_X(x)>0,\quad  x\in\mathcal{S},\\
	&\sum\limits_{y\in\mathcal{Y}}W(y|x)\log\frac{W(y|x)}{P^{(W)}_Y(y)}=C(W),\quad x\in\mathcal{S},\label{eq:capacity}\\
	&\sum\limits_{y\in\mathcal{Y}}W(y|x)\log\frac{W(y|x)}{P^{(W)}_Y(y)}<C(W),\quad x\in\mathcal{X}\backslash\mathcal{S},\label{eq:slackness3}
\end{align}
for some $\mathcal{S}\subseteq\mathcal{X}$ with $|\mathcal{S}|=|\mathcal{Y}|$,
where $C(W)$, $P^{(W)}_X$, and $P^{(W)}_Y$ denote the capacity, the capacity-achieving input distribution, and the capacity-achieving output distribution of $W$, respectively,
	
	\item (full rankness) the submatrix of 
	$\mathbb{W}$ consisting of the columns indexed by $\mathcal{S}$ has full rank.
\end{enumerate}
Correspondingly, we refer to $Q$ as a $W$-centered $\delta$-approximately memoryless surjective POST channel if \eqref{eq:infinity_norm} holds for a  memoryless surjective reference channel $W$.

The strict-complementary-slackness part of the surjectivity condition is essentially a strengthening of the standard Karush–Kuhn–Tucker conditions for the capacity-achieving input distribution of $W$. Indeed,  the latter can be stated as \cite[Theorem 4.5.1]{Gallager68}
\begin{align}
	&\sum\limits_{y\in\mathcal{Y}}W(y|x)\log\frac{W(y|x)}{P^{(W)}_Y(y)}=C(W),\quad x\in\mathcal{X} \mbox{ with }P^{(W)}_X(x)>0,\\
	&\sum\limits_{y\in\mathcal{Y}}W(y|x)\log\frac{W(y|x)}{P^{(W)}_Y(y)}\leq C(W),\quad x\in\mathcal{X} \mbox{ with }P^{(W)}_X(x)=0.\label{eq:notstrict}
\end{align}
In addition, there exists a subset $\mathcal{S}\subseteq\mathcal{X}$ with $|\mathcal{S}|\leq|\mathcal{Y}|$ \cite[Corollary 3, p. 96]{Gallager68} such that
\begin{align}
P^{(W)}_X(x)>0,\quad x\in\mathcal{S}.
\end{align}
Strict complementary slackness strengthens these requirements by enforcing 
$|\mathcal{S}|=|\mathcal{Y}|$ (instead of $|\mathcal{S}|\leq |\mathcal{Y}|$) and replacing \eqref{eq:notstrict} with a strict inequality.

The surjectivity condition necessarily requires $|\mathcal{X}|\geq |\mathcal{Y}|$. Moreover, since the capacity-achieving output distribution $P^{(W)}_Y$ is unique \cite[Corollary 2, p. 96]{Gallager68}, this condition  ensures that the capacity-achieving input distribution $P^{(W)}_X$ is uniquely determined\footnote{In light of \eqref{eq:slackness3}, we have $P^{(W)}_X(x)=0$ for $x\in\mathcal{X}\backslash\mathcal{S}$. Moreover, the full-rankness part of the surjectivity condition ensures that
the values of $P^{(W)}_X(x)$ for $x\in\mathcal{S}$ are uniquely determined by $P_Y^{(W)}$.} and assigns positive probability only to symbols in $\mathcal{S}$. Evidently,  
almost all memoryless channels with the input alphabet size at least as large as the output alphabet size satisfy this surjectivity condition, because the set of channels violating it lies in a lower-dimensional subspace of the full channel space.
The rationale for this condition and its technical implications will be clarified in Section \ref{sec:discussion}.

We now present our  main result, which demonstrates that the capacity of an approximately memoryless surjective POST channel is not enhanced by the presence of feedback.

\begin{theorem}\label{thm:main}
Given any  memoryless surjective channel $W$, there exists $\delta>0$ such that
\begin{align}
	C(Q)=C_f(Q)\label{eq:c=cf1}
\end{align}
for all $W$-centered $\delta$-approximately memoryless POST channels $Q$.
\end{theorem}
\begin{IEEEproof}
	See Section \ref{sec:proof}.
	\end{IEEEproof}

%\begin{theorem}\label{thm:extension}
%Given any surjective memoryless channel $W$ with $|\mathcal{X}|>|\mathcal{Y}|$, there exists $\delta>0$ such that
%\begin{align}
%	C(Q)=C_f(Q)\label{eq:c=cf2}
%\end{align}
%for all $W$-centered $\delta$-approximately memoryless POST channels $Q$.
%\end{theorem}
%\begin{IEEEproof}
%	See Section \ref{sec:proof2}.
%	\end{IEEEproof}

Theorem \ref{thm:main}  should not be misinterpreted as implying that memory has no effect on the channel capacity. In fact, for a 
$W$-centered 
$\delta$-approximately memoryless surjective POST channel 
$Q$, both 
$C(Q)$ and 
$C_f(Q)$ generally differ from 
$C(W)$. It merely indicates that whatever impact the channel memory has on the capacity, this impact is identical with and without feedback.

\section{Proof of the Main Result}\label{sec:proof}

The proof of Theorem \ref{thm:main} proceeds in two steps. We first 
establish the desired result for the case $|\mathcal{X}|=|\mathcal{Y}|$,  then show that the case 
$|\mathcal{X}|>|\mathcal{Y}|$ can be reduced to the equal-alphabet case, thereby completing the proof.

\subsection{$|\mathcal{X}|=|\mathcal{Y}|$}\label{sec:proof1}

First observe that when $|\mathcal{X}|=|\mathcal{Y}|$, the surjectivity  condition simplifies to 
\begin{enumerate}
	\item  (positivity) $P^{(W)}_X(x)>0$, $x\in\mathcal{X}$,
	\item   (full rankness) $\mathbb{W}$ has full rank.
\end{enumerate}
By Weyl's inequality for singular values, 
\begin{align}
	\sigma_{\min}(\mathbb{W})-\sigma_{\min}(\mathbb{Q}^{(c)}_{y'})\leq\sigma_{\max}(\mathbb{W}-\mathbb{Q}^{(c)}_{y'}),
\end{align}
where $\sigma_{\min}(\cdot)$ and $\sigma_{\max}(\cdot)$ denote the minimum and maximum singular values, respectively. Since
\begin{align}
	\sigma_{\max}(\mathbb{W}-\mathbb{Q}^{(c)}_{y'})&=\|\mathbb{W}-\mathbb{Q}^{(c)}_{y'}\|_2\nonumber\\
	&\leq \|\mathbb{W}-\mathbb{Q}^{(c)}_{y'}\|_F\nonumber\\
	&=\sqrt{\sum\limits_{x\in\mathcal{X}}\sum\limits_{y\in\mathcal{Y}}(\mathbb{W}(y,x)-\mathbb{Q}_{y'}^{(c)}(y,x))^2}\nonumber\\
	&\leq|\mathcal{X}|\|\mathbb{W}-\mathbb{Q}^{(c)}_{y'}\|_{\infty},
\end{align}
choosing
\begin{align}
	\delta<\frac{\sigma_{\min}(\mathbb{W})}{|\mathcal{X}|}\label{eq:full_rank}
\end{align}
guarantees that every channel transition matrix $\mathbb{Q}^{(c)}_{y'}$ associated with a $W$-centered $\delta$-approximately memoryless POST channel $Q$ satisfies $\sigma_{\min}(\mathbb{Q}^{(c)}_{y'})>0$, and is thus full rank for all
$y'\in\mathcal{Y}$.
Moreover, in light of Lemma \ref{lem:indecomposable_connected}, choosing 
\begin{align}
\delta<\min\left\{\min\limits_{x\in\mathcal{X}}\max\limits_{y\in\mathcal{Y}} W(y|x),\min\limits_{y\in\mathcal{Y}}\max_{x\in\mathcal{X}}W(y|x)\right\}\label{eq:minmax}
\end{align}
 ensures the applicability of the capacity characterizations in \eqref{eq:nonfeedback_capacity} and \eqref{eq:feedback_capacity}. Note that
 \begin{align}
 \min\limits_{x\in\mathcal{X}}\max\limits_{y\in\mathcal{Y}} W(y|x)\geq\frac{1}{|\mathcal{Y}|}>0,
 \end{align}
and
 \begin{align}
 	\min\limits_{y\in\mathcal{Y}}\max_{x\in\mathcal{X}}W(y|x)>0
 	\end{align}
 since otherwise there would exist some $y^*\in\mathcal{Y}$ with $W(y^*|x)=0$ for all $x\in\mathcal{X}$, which is ruled out by the full-rankness part of the simplified surjectivity condition. In the sequel, we assume $\delta$ satisfies both \eqref{eq:full_rank} and \eqref{eq:minmax}.

 The next two lemmas characterize the limiting behavior of   the maximizer  of the  optimization problem in \eqref{eq:feedback_capacity}--\eqref{eq:feedback_constraint}.
 
\begin{lemma}\label{lem:limit}
When the POST channel $Q$ degenerates to a memoryless surjective channel $W$, the optimization 
problem in \eqref{eq:feedback_capacity}--\eqref{eq:feedback_constraint} admits a unique maximizer, which is given by the product distribution
$P^{(W)}_XP^{(W)}_Y$.
\end{lemma}
\begin{IEEEproof}
		Observe that
	\begin{align}
		I(X;Y|Y')&=\sum\limits_{y'\in\mathcal{Y}}P_{Y'}(y')I(X;Y|Y'=y')\nonumber\\
		&\leq\sum\limits_{y'\in\mathcal{Y}}P_{Y'}(y')C(W)\nonumber\\
		&=C(W).
	\end{align}
	 Since the capacity-achieving input distribution for the memoryless surjective channel $W$ is uniquely given by $P^{(W)}_X$, the above upper bound is attained if and only if $P_{XY'}=P^{(W)}_XP_{Y'}$. 
	The corresponding output distribution is then 
	 $P_Y=P^{(W)}_Y$. Invoking the constraint \eqref{eq:feedback_constraint} yields $P_{Y'}=P^{(W)}_Y$, completing the proof.
	\end{IEEEproof}
 
\begin{lemma}\label{lem:uniform_convergence}
Given a  memoryless surjective channel $W$ with $|\mathcal{X}|=|\mathcal{Y}|$, the maximizer $P^*_{XY'}$ of the  optimization problem in \eqref{eq:feedback_capacity}--\eqref{eq:feedback_constraint} is unique for any $W$-centered $\delta$-approximately memoryless POST channel $Q$ when $\delta$ is sufficiently small, and it converges uniformly to $P^{(W)}_XP^{(W)}_Y$ over all such POST channels as $\delta\rightarrow 0$.
\end{lemma}
\begin{IEEEproof}
See Appendix \ref{app:uniform_convergence}.
	\end{IEEEproof}

By the positivity part of the simplified surjectivity condition, the capacity-achieving input distribution satisfies $P^{(W)}_X(x)>0$ for all $x\in\mathcal{X}$. The full-rankness part of the simplified surjectivity condition  further guarantees that $\max_{x\in\mathcal{X}}W(y|x)>0$ for all $y\in\mathcal{Y}$, which in turn ensures that the corresponding output distribution satisfies $P^{(W)}_Y(y)>0$ for all $y\in\mathcal{Y}$. Consequently, by Lemma \ref{lem:uniform_convergence}, the marginal distributions 
 $P^*_X$ and $P^*_{Y'}$ induced by the maximizer $P^*_{XY'}$ satisfy $P^*_X(x)>0$ for all $x\in\mathcal{X}$ and $P^*_{Y'}(y)>0$ for $y\in\mathcal{Y}$, provided that $\delta$ is sufficiently small.

%Let $P^*_{XY'}$ be a maximizer of the optimization problem in \eqref{eq:feedback_capacity}--\eqref{eq:feedback_constraint} and let $P_{YY'}$ be the joint distribution of $(Y,Y')$ induced through \eqref{eq:feedback_factorization}.

Construct a stationary Markov process $\{Y_t\}_{t=0}^{\infty}$ such that $P_{Y_0}=P^*_{Y'}$ and $P_{Y_t|Y_{t-1}}=P^*_{Y|Y'}$ for $t=1,2,\ldots$, where $P^*_{XYY'}$ is induced by the maximizer $P^*_{XY'}$ via \eqref{eq:feedback_factorization}.
This process can be viewed as the output sequence of the POST channel $Q$ generated by the optimal feedback strategy that achieves  $C_f(Q)$.
The crux of the proof is to show that there exists an input distribution capable of simulating this Markov process through $Q$ 
 without feedback.
To this end, define
\begin{align}
P_{X^nY^nY_0}(x^n,y^n,y_0):=	P_{X^n|Y_0}(x^n|y_0)P^*_{Y'}(y_0)\left(\prod\limits_{t=1}^nQ(y_t|x_t,y_{t-1})\right).\label{eq:nonfeedback_POST}
	\end{align}
Our goal is to demonstrate that, for every positive integer $n$, one can find a conditional input distribution $P_{X^n|Y_0}$ satisfying
\begin{align}
	P_{Y^nY_0}(y^n,y_0):=\sum\limits_{x^n\in\mathcal{X}^n} P_{X^nY^nY_0}(x^n,y^n,y_0)=P^*_{Y'}(y_0)\prod\limits_{t=1}^nP^*_{Y|Y'}(y_t|y_{t-1}).\label{eq:Markov}
\end{align}
Given such a $P_{X^n|Y_0}$, we obtain
under  $P_{X^nY^nY_0}$
\begin{align}
I(X^n;Y^n|Y_0)&=H(Y^n|Y_0)-H(Y^n|X^n,Y_0)\nonumber\\
&=\sum\limits_{t=1}^nH(Y_t|Y^{t-1}_0)-\sum\limits_{t=1}^nH(Y_t|X^n,Y^{t-1}_0)\nonumber\\
&\stackrel{(a)}{=}\sum\limits_{t=1}^nH(Y_t|Y^{t-1}_0)-\sum\limits_{t=1}^nH(Y_t|X_t,Y_{t-1})\nonumber\\
&\stackrel{(b)}{=}nH(Y|Y')-\sum\limits_{t=1}^nH(Y_t|X_t,Y_{t-1}),\label{eq:sub}
	\end{align}
where ($a$) and ($b$) are due to \eqref{eq:nonfeedback_POST} and \eqref{eq:Markov}, respectively. Since all channel transition matrices $\mathbb{Q}^{(c)}_0,\mathbb{Q}^{(c)}_1,\ldots,\mathbb{Q}^{(c)}_{|\mathcal{Y}|-1}$ have full rank, the fact that 
$P_{Y_t|Y_{t-1}}=P^*_{Y|Y'}$  forces $P_{X_t|Y_{t-1}}$
to match $P^*_{X|Y'}$ induced by the maximizer $P^*_{XY'}$.
Hence,
\begin{align}
	H(Y_t|X_t,Y_{t-1})=H(Y|X,Y')\label{eq:XY'_match}
\end{align}
for $t=1,2,\ldots,n$.
Substituting \eqref{eq:XY'_match} into
\eqref{eq:sub} yields
\begin{align}
I(X^n;Y^n|Y_0)=nI(X;Y|Y')=nC_f(Q).\label{eq:comb1}
\end{align}
On the other hand,
\begin{align}
I(X^n;Y^n|Y_0)\leq nC_n(Q),\label{eq:comb2}
\end{align}
where $C_n(Q)$ is defined in \eqref{eq:C_n(Q)}. 
In view of \eqref{eq:nonfeedback_capacity}, \eqref{eq:comb1}, and \eqref{eq:comb2}, sending $n\rightarrow\infty$ proves $C(Q)\geq C_f(Q)$. Since the converse inequality holds trivially, we conclude that $C(Q)=C_f(Q)$.

It remains to establish the existence of such a $P_{X^n|Y_0}$. At a high level, the strategy used to prove
 $C(Q)=C_f(Q)$ mirrors that of \cite{PAW14}: we aim to simulate, through the POST channel $Q$ without feedback, the Markov process generated by the optimal feedback strategy. The crucial distinction is that the specific POST$(a,b)$ channel studied in \cite{PAW14} admits closed-form expressions for all relevant quantities, reducing the proof to checking some explicit inequalities. In the general setting considered here, such explicit evaluation is no longer feasible. As a result, our proof must take a more conceptual route, relying on structural properties of approximately memoryless surjective POST channels rather than direct calculations.

For $n=1,2,\ldots$, define
\begin{align}
	&Q^{(n)}(y^n|x^n,y_0):=\prod\limits_{t=1}^nQ(y_t|x_t,y_{t-1}),\label{eq:recursion1}\\
	&P^{(n)}(y^n|y_0):=\prod\limits_{t=1}^n P^*_{Y|Y'}(y_t|y_{t-1}).\label{eq:recursion2}
\end{align}
Note that $Q^{(n)}(y^n|x^n,y_0)$ and $P^{(n)}(y^n|y_0)$ admit the following recursion:
\begin{align}
	&Q^{(n)}(y^n|x^n,y_0)=Q(y_1|x_1,y_0)Q^{(n-1)}(y^n_2|x^n_2,y_1),\\
	&P^{(n)}(y^n|y_0)=P^*_{Y|Y'}(y_1|y_0)P^{(n-1)}(y^n_2|y_1).
\end{align}
Construct\footnote{When sequences such as $x^n$ or $y^n$ are used to index the columns or rows of a matrix, they are interpreted according to their base-$|\mathcal{X}|$
or base-$|\mathcal{Y}|$
 representations, respectively. For instance, if 
$|\mathcal{X}|=|\mathcal{Y}|=2$, 
$x^n=(1,0,1)$, and 
$y^n=(1,1,0)$, then 
$\mathbb{Q}(y^n,x^n)$ denotes the 
$(6,5)$-entry of 
$\mathbb{Q}$, i.e., 
$\mathcal{Q}(6,5)$, since these sequences correspond to the base-$2$ indices 
$6$ and 
$5$. As a reminder, matrix rows and columns are indexed starting from $0$.
} matrix $\mathbb{Q}^{(n)}_{y_0}$ and vector $\mathbbm{q}^{(n)}_{y_0}$ such that $\mathbb{Q}^{(n)}_{y_0}(y^n,x^n)=Q^{(n)}(y^n|x^n,y_0)$ and $\mathbbm{q}^{(n)}_{y_0}(y^n)=P^{(n)}(y^n|y_0)$. The problem of finding a conditional input distribution $P_{X^n|Y_0}$ meeting the output constraint \eqref{eq:Markov}
is equivalent to proving the existence of a probability vector $\mathbbm{p}^{(n)}_{y_0}$ satisfying
\begin{align}
	\mathbb{Q}^{(n)}_{y_0}\mathbbm{p}^{(n)}_{y_0}=\mathbbm{q}^{(n)}_{y_0}
\end{align}
for every $y_0\in\mathcal{Y}$.

In view of \eqref{eq:recursion1} and \eqref{eq:recursion2}, we have
\begin{align}
	&\mathbb{Q}^{(n)}_{y_0}=\left(\begin{matrix}
	\mathbb{Q}^{(c)}_{y_0}(0,0)\mathbb{Q}^{(n-1)}_{0} & \mathbb{Q}^{(c)}_{y_0}(0,1)\mathbb{Q}^{(n-1)}_{0} & \cdots & \mathbb{Q}^{(c)}_{y_0}(0,|\mathcal{X}|-1)\mathbb{Q}^{(n-1)}_{0}\\
	\mathbb{Q}^{(c)}_{y_0}(1,0)\mathbb{Q}^{(n-1)}_{1} & \mathbb{Q}^{(c)}_{y_0}(1,1)\mathbb{Q}^{(n-1)}_{1} &\cdots & \mathbb{Q}^{(c)}_{y_0}(1,|\mathcal{X}|-1)\mathbb{Q}^{(n-1)}_{1}\\
	\vdots & \vdots & \ddots & \vdots\\
	\mathbb{Q}^{(c)}_{y_0}(|\mathcal{Y}|-1,0)\mathbb{Q}^{(n-1)}_{|\mathcal{Y}|-1} & \mathbb{Q}^{(c)}_{y_0}(|\mathcal{Y}|-1,1)\mathbb{Q}^{(n-1)}_{|\mathcal{Y}|-1} & \cdots & \mathbb{Q}^{(c)}_{y_0}(|\mathcal{Y}|-1||\mathcal{X}|-1)\mathbb{Q}^{(n-1)}_{|\mathcal{Y}|-1} \end{matrix}\right),\label{eq:recursionmatrix1}\\
	&\mathbbm{q}^{(n)}_{y_0}=\left(\begin{matrix}
		\mathbbm{q}^*_{y_0}(0)\mathbbm{q}^{(n-1)}_{0}\\
		\mathbbm{q}^*_{y_0}(1)\mathbbm{q}^{(n-1)}_{1}\\
		\vdots\\
		\mathbbm{q}^*_{y_0}(|\mathcal{Y}|-1)\mathbbm{q}^{(n-1)}_{|\mathcal{Y}|-1}
		\end{matrix}\label{eq:recursionmatrix2}
	\right),
\end{align}
where $\mathbbm{q}^*_{y'}:=(P^*_{Y|Y'}(y|y'))^T_{y\in\mathcal{Y}}$ for $y'\in\mathcal{Y}$.
Recall that all channel transition matrices $\mathbbm{Q}^{(c)}_{0},\mathbbm{Q}^{(c)}_{1},\ldots,\mathbbm{Q}^{(c)}_{|\mathcal{Y}|-1}$ have full rank whenever \eqref{eq:full_rank} holds. Let
\begin{align} \mathbb{G}_{y'}:=(\mathbb{Q}^{(c)}_{y'})^{-1}\label{eq:inverse_base}
	\end{align}
	 for $y'\in\mathcal{Y}$. Recursively define
\begin{align}
&\mathbb{G}^{(n)}_{y_0}:=\left(\begin{matrix}
	\mathbb{G}_{y_0}(0,0)\mathbb{G}^{(n-1)}_{0} & \mathbb{G}_{y_0}(0,1)\mathbb{G}^{(n-1)}_{1} & \cdots & \mathbb{G}_{y_0}(0,|\mathcal{Y}|-1)\mathbb{G}^{(n-1)}_{|\mathcal{Y}|-1}\\
	\mathbb{G}_{y_0}(1,0)\mathbb{G}^{(n-1)}_{0} & \mathbb{G}_{y_0}(1,1)\mathbb{G}^{(n-1)}_{1} &\cdots & \mathbb{G}_{y_0}(1,|\mathcal{Y}|-1)\mathbb{G}^{(n-1)}_{|\mathcal{Y}|-1}\\
	\vdots & \vdots & \ddots & \vdots\\
	\mathbb{G}_{y_0}(|\mathcal{X}|-1,0)\mathbb{G}^{(n-1)}_{0} & \mathbb{G}_{y_0}(|\mathcal{X}|-1,1)\mathbb{G}^{(n-1)}_{1} & \cdots & \mathbb{G}_{y_0}(|\mathcal{X}|-1|,|\mathcal{Y}|-1)\mathbb{G}^{(n-1)}_{|\mathcal{Y}|-1} \end{matrix}\right)\label{eq:recursionmatrix3}
\end{align}
with $\mathbb{G}^{(1)}_{y_0}:=\mathbb{G}_{y_0}$ for $y_0\in\mathcal{Y}$. We shall show via induction that
\begin{align}
	\mathbb{G}^{(n)}_{y_0}=(\mathbb{Q}^{(n)}_{y_0})^{-1}.\label{eq:inverse}
\end{align}
Clearly, \eqref{eq:inverse} holds when $n=1$ since
\begin{align}
\mathbb{G}^{(1)}_{y_0}=\mathbb{G}_{y_0}=(\mathbb{Q}^{(c)}_{y_0})^{-1}=(\mathbb{Q}^{(1)}_{y_0})^{-1}.
\end{align}
Assume $\mathbb{G}^{(n-1)}_{y_0}=(\mathbb{Q}^{(n-1)}_{y_0})^{-1}$. It can be verified by leveraging \eqref{eq:recursionmatrix1}, \eqref{eq:inverse_base}, and \eqref{eq:recursionmatrix3}
that
\begin{align}
	&\mathbb{G}^{(n)}_{y_0}\mathbb{Q}^{(n)}_{y_0}=\left(\begin{matrix}
		(\mathbb{G}_{y_0}\mathbb{Q}^{(c)}_{y_0})(0,0)\mathbb{I} & (\mathbb{G}_{y_0}\mathbb{Q}^{(c)}_{y_0})(0,1)\mathbb{I} & \cdots & (\mathbb{G}_{y_0}\mathbb{Q}^{(c)}_{y_0})(0,|\mathcal{X}|-1)\mathbb{I}\\
		(\mathbb{G}_{y_0}\mathbb{Q}^{(c)}_{y_0})(1,0)\mathbb{I} & (\mathbb{G}_{y_0}\mathbb{Q}^{(c)}_{y_0})(1,1)\mathbb{I} &\cdots & (\mathbb{G}_{y_0}\mathbb{Q}^{(c)}_{y_0})(1,|\mathcal{X}|-1)\mathbb{I}\\
		\vdots & \vdots & \ddots & \vdots\\
		(\mathbb{G}_{y_0}\mathbb{Q}^{(c)}_{y_0})(|\mathcal{X}|-1,0)\mathbb{I} & (\mathbb{G}_{y_0}\mathbb{Q}^{(c)}_{y_0})(|\mathcal{X}|-1,1)\mathbb{I} & \cdots & (\mathbb{G}_{y_0}\mathbb{Q}^{(c)}_{y_0})(|\mathcal{X}|-1|,|\mathcal{X}|-1)\mathbb{I} \end{matrix}\right)=\mathbb{I},\\
			&\mathbb{Q}^{(n)}_{y_0}\mathbb{G}^{(n)}_{y_0}=\left(\begin{matrix}
			(\mathbb{Q}^{(c)}_{y_0}\mathbb{G}_{y_0})(0,0)\mathbb{I} & (\mathbb{Q}^{(c)}_{y_0}\mathbb{G}_{y_0})(0,1)\mathbb{I} & \cdots & (\mathbb{Q}^{(c)}_{y_0}\mathbb{G}_{y_0})(0,|\mathcal{Y}|-1)\mathbb{I}\\
			(\mathbb{Q}^{(c)}_{y_0}\mathbb{G}_{y_0})(1,0)\mathbb{I} & (\mathbb{Q}^{(c)}_{y_0}\mathbb{G}_{y_0})(1,1)\mathbb{I} &\cdots & (\mathbb{Q}^{(c)}_{y_0}\mathbb{G}_{y_0})(1,|\mathcal{Y}|-1)\mathbb{I}\\
			\vdots & \vdots & \ddots & \vdots\\
			(\mathbb{Q}^{(c)}_{y_0}\mathbb{G}_{y_0})(|\mathcal{Y}|-1,0)\mathbb{I} & (\mathbb{Q}^{(c)}_{y_0}\mathbb{G}_{y_0})(|\mathcal{Y}|-1,1)\mathbb{I} & \cdots & (\mathbb{Q}^{(c)}_{y_0}\mathbb{G}_{y_0})(|\mathcal{Y}|-1|,|\mathcal{Y}|-1)\mathbb{I} \end{matrix}\right)=\mathbb{I},
\end{align}
thereby establishing \eqref{eq:inverse}. Now the problem boils down to proving 
\begin{align}
	\mathbbm{p}^{(n)}_{y_0}:=\mathbb{G}^{(n)}_{y_0}\mathbbm{q}^{(n)}_{y_0}\label{eq:inputvector}
\end{align}
is a valid probability vector.

Note that
\begin{align}
	\mathbbm{1}^T\mathbbm{p}^{(n)}_{y_0}&=\mathbbm{1}^T\mathbb{G}^{(n)}_{y_0}\mathbbm{q}^{(n)}_{y_0}\nonumber\\
	&\stackrel{(a)}{=}\mathbbm{1}^T\mathbb{Q}^{(n)}_{y_0}\mathbb{G}^{(n)}_{y_0}\mathbbm{q}^{(n)}_{y_0}\nonumber\\
	&\stackrel{(b)}{=}\mathbbm{1}^T\mathbbm{q}^{(n)}_{y_0}\nonumber\\
	&=1,\label{eq:normalization}
\end{align}
where ($a$) and ($b$) are due to $\mathbbm{1}^T\mathbb{Q}^{(n)}_{y_0}=\mathbbm{1}$ and \eqref{eq:inverse}, respectively. It remains to ensure that
\begin{align}
	\mathbbm{p}^{(n)}_{y_0}(x^n)\geq 0\label{eq:convex_hull}
\end{align}
for all $x^n\in\mathcal{X}^n$. 
As certain steps in our argument require $\mathbbm{p}^{(n)}_{y_0}(x^n)$ to appear in the denominator, to facilitate the induction, we therefore impose the stronger condition that 
\begin{align} \mathbbm{p}^{(n)}_{y_0}(x^n)>0.\label{eq:positivity}
\end{align}
 By Lemma \ref{lem:uniform_convergence},
$P^*_{XY}$ converges uniformly to $P^{(W)}_XP^{(W)}_Y$, which in turn implies
\begin{align}
	&P^*_{X|Y'}(x|y')\rightarrow P^{(W)}_X(x),\quad x\in\mathcal{X}, y'\in\mathcal{Y},\label{eq:convergence_X}\\
	&P^*_{Y|Y'}(y|y')\rightarrow P^{(W)}_Y(y),\quad y,y'\in\mathcal{Y},\label{eq:convergence_Y}
\end{align}
as $\delta\rightarrow0$. Note that
\begin{align}
	\mathbbm{p}^{(1)}_{y_0}=\mathbb{G}_{y_0}\mathbbm{q}^*_{y_0}=(\mathbb{Q}^{(c)}_{y_0})^{-1}\mathbbm{q}^*_{y_0}=\mathbbm{p}^*_{y_0}\rightarrow\mathbbm{p}^{(W)}\label{eq:input_convergence}
\end{align}
as $\delta\rightarrow 0$, where $\mathbbm{p}^*_{y'}:=(P^*_{X|Y'}(x|y'))^T_{x\in\mathcal{X}}$
for $y'\in\mathcal{Y}$ and $\mathbbm{p}^{(W)}:=(p^{(W)}_X(x))^T_{x\in\mathcal{X}}$. Therefore, \eqref{eq:positivity} holds when $n=1$ for all sufficiently small $\delta$. Assume
\begin{align}
	\mathbbm{p}^{(n-1)}_{y_0}(x^{n-1})>0\label{eq:com1}
\end{align}
for all $x^{n-1}\in\mathcal{X}^{n-1}$.
It follows by \eqref{eq:recursionmatrix2}, \eqref{eq:recursionmatrix3}, and \eqref{eq:inputvector} that
\begin{align}
	\mathbbm{p}^{(n)}_{y_0}(x^n)=\sum\limits_{y_1\in\mathcal{Y}}\mathbb{G}_{y_0}(x_1,y_1)\mathbbm{q}^*_{y_0}(y_1)\mathbbm{p}^{(n-1)}_{y_1}(x^n_2).\label{eq:expansion2}
\end{align}
Therefore, we have
\begin{align}
	\frac{\mathbbm{p}^{(n)}_{y_0}(x^n)}{\mathbbm{p}^{(n-1)}_{y_0}(x^{n}_2)}&=\sum\limits_{y_1\in\mathcal{Y}}\mathbb{G}_{y_0}(x_1,y_1)\mathbbm{q}^*_{y_0}(y_1)\frac{\mathbbm{p}^{(n-1)}_{y_1}(x^n_2)}{\mathbbm{p}^{(n-1)}_{y_0}(x^{n}_2)}\nonumber\\
	&=\sum\limits_{y_1\in\mathcal{Y}}\mathbb{G}_{y_0}(x_1,y_1)\mathbbm{q}^*_{y_0}(y_1)+\sum\limits_{y_1\in\mathcal{Y}}\mathbb{G}_{y_0}(x_1,y_1)\mathbbm{q}^*_{y_0}(y_1)\frac{\mathbbm{p}^{(n-1)}_{y_1}(x^n_2)-\mathbbm{p}^{(n-1)}_{y_0}(x^{n}_2)}{\mathbbm{p}^{(n-1)}_{y_0}(x^{n}_2)}\nonumber\\
	&\stackrel{(c)}{=}\mathbbm{p}^*_{y_0}(x_1)+\sum\limits_{y_1\in\mathcal{Y}}\mathbb{G}_{y_0}(x_1,y_1)\mathbbm{q}^*_{y_0}(y_1)\frac{\mathbbm{p}^{(n-1)}_{y_1}(x^n_2)-\mathbbm{p}^{(n-1)}_{y_0}(x^{n}_2)}{\mathbbm{p}^{(n-1)}_{y_0}(x^{n}_2)}\nonumber\\
	&\geq\mathbbm{p}^*_{y_0}(x_1)-\sum\limits_{y_1\in\mathcal{Y}}\left|\mathbb{G}_{y_0}(x_1,y_1)\right|\mathbbm{q}^*_{y_0}(y_1)\left|\frac{\mathbbm{p}^{(n-1)}_{y_1}(x^n_2)-\mathbbm{p}^{(n-1)}_{y_0}(x^{n}_2)}{\mathbbm{p}^{(n-1)}_{y_0}(x^{n}_2)}\right|,\label{eq:com2}
\end{align}
where ($c$) is due to \eqref{eq:input_convergence}. 
By \eqref{eq:convergence_X}, \eqref{eq:convergence_Y}, and the fact that  $\mathbb{G}_{y_0}\rightarrow\mathbb{W}^{-1}$ as $\delta\rightarrow 0$,
\begin{align}
	\lim\limits_{\delta\rightarrow0}\frac{\mathbbm{p}^*_{y_0}(x_1)}{\sum_{y_1\in\mathcal{Y}}\left|\mathbb{G}_{y_0}(x_1,y_1)\right|\mathbbm{q}^*_{y_0}(y_1)}=\frac{\mathbbm{p}^{(W)}(x_1)}{\sum_{y\in\mathcal{Y}}|\mathbb{W}^{-1}(x_1,y_1)|\mathbbm{q}^{(W)}(y_1)}.
\end{align}
Therefore, if there exists $\kappa$ satisfying
\begin{align}
	0<\kappa<\min\limits_{x\in\mathcal{X}}\frac{\mathbbm{p}^{(W)}(x)}{\sum_{y\in\mathcal{Y}}|\mathbb{W}^{-1}(x,y)|\mathbbm{q}^{(W)}(y)}\label{eq:kappa}
\end{align}
such that
\begin{align}
	\left|\frac{\mathbbm{p}^{(n)}_{y'_0}(x^n)-\mathbbm{p}^{(n)}_{y_0}(x^n)}{\mathbbm{p}^{(n)}_{y_0}(x^n)}\right|\leq\kappa\label{eq:no_sign_change}
\end{align}
for all $n$, $x^n\in\mathcal{X}^n$, and $y_0, y'_0\in\mathcal{Y}$, then we can find a small enough $\delta$ ensuring
\begin{align}
	\frac{\mathbbm{p}^*_{y_0}(x_1)}{\sum_{y_1\in\mathcal{Y}}\left|\mathbb{G}_{y_0}(x_1,y_1)\right|\mathbbm{q}^*_{y_0}(y_1)}>\kappa
\end{align}
and consequently
\begin{align}
	\mathbbm{p}^*_{y_0}(x_1)-\sum\limits_{y_1\in\mathcal{Y}}\left|\mathbb{G}_{y_0}(x_1,y_1)\right|\mathbbm{q}^*_{y_0}(y_1)\left|\frac{\mathbbm{p}^{(n-1)}_{y_1}(x^n_2)-\mathbbm{p}^{(n-1)}_{y_0}(x^{n}_2)}{\mathbbm{p}^{(n-1)}_{y_0}(x^{n}_2)}\right|>0.\label{eq:com3}
\end{align}
Combining \eqref{eq:com2} and \eqref{eq:com3} with the induction hypothesis \eqref{eq:com1} proves \eqref{eq:positivity}.

%when $\delta$ is sufficiently small, there exists $\kappa\in(0,1)$ such  As a consequence, the sign of $\mathbbm{p}^{(n)}_{y_0}(x^n)$ does not depend on $x^n$, which, together with \eqref{eq:normalization}, yields \eqref{eq:positivity}.

Finally, we shall show via induction
that for any $\kappa$ satisfying \eqref{eq:kappa}, when $\delta$ is sufficiently small, \eqref{eq:no_sign_change} holds for all $n$, $x^n\in\mathcal{X}^n$, and $y_0, y'_0\in\mathcal{Y}$. In light of \eqref{eq:input_convergence}, this is true when $n=1$.
Assume
\begin{align}
\left|\frac{\mathbbm{p}^{(n-1)}_{y'_0}(x^{n-1})-\mathbbm{p}^{(n-1)}_{y_0}(x^{n-1})}{\mathbbm{p}^{(n-1)}_{y_0}(x^{n-1})}\right|\leq\kappa.\label{eq:induction_hypothesis}
\end{align}
By \eqref{eq:expansion2}, we have
\begin{align}
	\left|\frac{\mathbbm{p}^{(n)}_{y'_0}(x^n)-\mathbbm{p}^{(n)}_{y_0}(x^n)}{\mathbbm{p}^{(n)}_{y_0}(x^n)}\right|&=\left|\frac{\sum_{y_1\in\mathcal{Y}}\mathbb{G}_{y'_0}(x_1,y_1)\mathbbm{q}^*_{y'_0}(y_1)\mathbbm{p}^{(n-1)}_{y_1}(x^n_2)-\sum\limits_{y_1\in\mathcal{Y}}\mathbb{G}_{y_0}(x_1,y_1)\mathbbm{q}^*_{y_0}(y_1)\mathbbm{p}^{(n-1)}_{y_1}(x^n_2)}{\sum_{y_1\in\mathcal{Y}}\mathbb{G}_{y_0}(x_1,y_1)\mathbbm{q}^*_{y_0}(y_1)\mathbbm{p}^{(n-1)}_{y_1}(x^n_2)}\right|\nonumber\\
	&\leq\frac{\sum_{y_1\in\mathcal{Y}}\left|(\mathbb{G}_{y'_0}(x_1,y_1)\mathbbm{q}^*_{y'_0}(y_1)-\mathbb{G}_{y_0}(x_1,y_1)\mathbbm{q}^*_{y_0}(y_1))\mathbbm{p}^{(n-1)}_{y_1}(x^n_2)\right|}{\left|\sum_{y_1\in\mathcal{Y}}\mathbb{G}_{y_0}(x_1,y_1)\mathbbm{q}^*_{y_0}(y_1)\mathbbm{p}^{(n-1)}_{y_1}(x^n_2)\right|}\nonumber\\
	&=\frac{\sum_{y_1\in\mathcal{Y}}\left|(\mathbb{G}_{y'_0}(x_1,y_1)\mathbbm{q}^*_{y'_0}(y_1)-\mathbb{G}_{y_0}(x_1,y_1)\mathbbm{q}^*_{y_0}(y_1))\mathbbm{\tau}^{(n-1)}_{y_1}(x^n_2)\right|}{\left|\sum_{y_1\in\mathcal{Y}}\mathbb{G}_{y_0}(x_1,y_1)\mathbbm{q}^*_{y_0}(y_1)\mathbbm{\tau}^{(n-1)}_{y_1}(x^n_2)\right|},\label{eq:deviation}
\end{align}
where
\begin{align}
	\mathbbm{\tau}^{(n-1)}_{y_1}(x^n_2):=\frac{\mathbbm{p}^{(n-1)}_{y_1}(x^n_2)}{\max_{y'_1\in\mathcal{Y}}\mathbbm{p}^{(n-1)}_{y'_1}(x^n_2)}.
\end{align}
In view of  the induction hypothesis \eqref{eq:induction_hypothesis},
\begin{align}
\sum_{y_1\in\mathcal{Y}}\left|(\mathbb{G}_{y'_0}(x_1,y_1)\mathbbm{q}^*_{y'_0}(y_1)-\mathbb{G}_{y_0}(x_1,y_1)\mathbbm{q}^*_{y_0}(y_1))\mathbbm{\tau}^{(n-1)}_{y_1}(x^n_2)\right|&\leq(1+\kappa)\sum_{y_1\in\mathcal{Y}}\left|(\mathbb{G}_{y'_0}(x_1,y_1)\mathbbm{q}^*_{y'_0}(y_1)-\mathbb{G}_{y_0}(x_1,y_1)\mathbbm{q}^*_{y_0}(y_1))\right|.\label{eq:numerator}
\end{align}
Moreover,
\begin{align}
	\sum_{y_1\in\mathcal{Y}}\mathbb{G}_{y_0}(x_1,y_1)\mathbbm{q}^*_{y_0}(y_1)\mathbbm{\tau}^{(n-1)}_{y_1}(x^n_2)&=	\sum_{y_1\in\mathcal{Y}}\mathbb{G}_{y_0}(x_1,y_1)\mathbbm{q}^*_{y_0}(y_1)+\sum_{y_1\in\mathcal{Y}}\mathbb{G}_{y_0}(x_1,y_1)\mathbbm{q}^*_{y_0}(y_1)(\mathbbm{\tau}^{(n-1)}_{y_1}(x^n_2)-1)\nonumber\\
	&\stackrel{(d)}{=}\mathbbm{p}^*_{y_0}(x_1)+\sum_{y_1\in\mathcal{Y}}\mathbb{G}_{y_0}(x_1,y_1)\mathbbm{q}^*_{y_0}(y_1)(\mathbbm{\tau}^{(n-1)}_{y_1}(x^n_2)-1)\nonumber\\
	&\stackrel{(e)}{\geq}\mathbbm{p}^*_{y_0}(x_1)-\kappa\sum_{y_1\in\mathcal{Y}}\left|\mathbb{G}_{y_0}(x_1,y_1)\right|\mathbbm{q}^*_{y_0}(y_1),\label{eq:denominator}
\end{align}
where ($d$) and ($e$) are due to \eqref{eq:input_convergence} and \eqref{eq:induction_hypothesis}, respectively.
By \eqref{eq:convergence_X}, \eqref{eq:convergence_Y}, and the fact that  $\mathbb{G}_{y_0}\rightarrow\mathbb{W}^{-1}$ as $\delta\rightarrow 0$,
\begin{align}
&\lim\limits_{\delta\rightarrow 0}\sum_{y_1\in\mathcal{Y}}\left|(\mathbb{G}_{y'_0}(x_1,y_1)\mathbbm{q}^*_{y'_0}(y_1)-\mathbb{G}_{y_0}(x_1,y_1)\mathbbm{q}^*_{y_0}(y_1))\right|=0,\\
&\lim\limits_{\delta\rightarrow 0}\mathbbm{p}^*_{y_0}(x_1)-\kappa\sum_{y_1\in\mathcal{Y}}\left|\mathbb{G}_{y_0}(x_1,y_1)\right|\mathbbm{q}^*_{y_0}(y_1)=\mathbbm{p}^{(W)}(x_1)-\kappa\sum\limits_{y_1\in\mathcal{Y}}\left|\mathbb{W}^{-1}(x_1,y_1)\right|\mathbbm{q}^{(W)}(y_1).
\end{align}
Hence, given $\kappa$ satisfying \eqref{eq:kappa},
one can find a small enough $\delta$ such that
\begin{align}
&\mathbbm{p}^*_{y_0}(x_1)-\kappa\sum_{y_1\in\mathcal{Y}}\left|\mathbb{G}_{y_0}(x_1,y_1)\right|\mathbbm{q}^*_{y_0}(y_1)> 0,\label{eq:cond1}\\
&\frac{(1+\kappa)\sum_{y_1\in\mathcal{Y}}\left|(\mathbb{G}_{y'_0}(x_1,y_1)\mathbbm{q}^*_{y'_0}(y_1)-\mathbb{G}_{y_0}(x_1,y_1)\mathbbm{q}^*_{y_0}(y_1))\right|}{\mathbbm{p}^*_{y_0}(x_1)-\kappa\sum_{y_1\in\mathcal{Y}}\left|\mathbb{G}_{y_0}(x_1,y_1)\right|\mathbbm{q}^*_{y_0}(y_1)}\leq\kappa.\label{eq:cond2}
\end{align}
Combining \eqref{eq:deviation}, \eqref{eq:numerator}, \eqref{eq:denominator}, \eqref{eq:cond1}, and \eqref{eq:cond2} establishes \eqref{eq:no_sign_change}.

\subsection{$|\mathcal{X}|>|\mathcal{Y}|$}\label{sec:proof2}

To address the case $|\mathcal{X}|>|\mathcal{Y}|$, a refined version of Lemma \ref{lem:uniform_convergence} is required.

\begin{lemma}\label{lem:refined_version}
	When $\delta$ is sufficiently small, for any $W$-centered $\delta$-approximately memoryless POST channel $Q$ with respect to a given memoryless surjective channel $W$, the maximizer $P^*_{XY'}$ of the  optimization problem in \eqref{eq:feedback_capacity}--\eqref{eq:feedback_constraint} is unique, and the induced marginal $P^*_X$ is supported on  $\mathcal{S}$; moreover, $P^*_{XY'}$ converges uniformly to $P^{(W)}_XP^{(W)}_Y$ over all such POST channels as $\delta\rightarrow 0$.
	\end{lemma}

\begin{IEEEproof}
	Using essentially the same argument as in the proof of Lemma \ref{lem:uniform_convergence}, one can show that $P^*_{XY'}$ converges uniformly to $P^{(W)}_XP^{(W)}_Y$ as $\delta\rightarrow 0$. 
	It remains to establish that, for sufficiently small
	$\delta$, the induced marginal $P^*_X$ is supported on $\mathcal{S}$, which in turn guarantees the uniqueness of  $P^*_{XY'}$ by Lemma \ref{lem:uniform_convergence}.

The uniform convergence of $P^*_{XY}$ to $P^{(W)}_XP^{(W)}_Y$ implies that the induced conditional distributions satisfy
\begin{align}
	&P^*_{X|Y'}(x|y')\rightarrow P^{(W)}_X(x),\quad x\in\mathcal{X}, y'\in\mathcal{Y},\\
	&P^*_{Y|Y'}(y|y')\rightarrow P^{(W)}_Y(y),\quad y,y'\in\mathcal{Y},
\end{align}
as $\delta\rightarrow0$. Now suppose, to the contrary, that $\sum_{x\in\mathcal{X}\backslash\mathcal{S}}P^*_X(x)>0$. 
For any $\mathcal{A}\subseteq\mathcal{X}$,
let $\mathbb{W}(\mathcal{A})$ denote the submatrix of $\mathbb{W}$ consisting of the columns indexed by $\mathcal{A}$, and define 
$\mathbb{Q}^{(c)}_{y'}(\mathcal{A})$ analogously for $y'\in\mathcal{Y}$. Note that the matrices $\mathbb{Q}^{(c)}_{0}(\mathcal{S}),\mathbb{Q}^{(c)}_{1}(\mathcal{S}),\ldots,\mathbb{Q}^{(c)}_{|\mathcal{Y}|-1}(\mathcal{S})$ have full rank whenever
\begin{align}
	\delta<\frac{\sigma_{\min}(\mathbb{W}(\mathcal{S}))}{|\mathcal{S}|}.
	\end{align}
%Without loss of generality, we assume $\mathcal{S}=\{0,1,\ldots,|\mathcal{S}|-1\}$. 
Construct a new conditional distribution $P'_{X|Y'}$ such that
\begin{align}
&\mathbbm{p}'_{y'}(\mathcal{S})=(\mathbb{Q}^{(c)}_{y'}(\mathcal{S}))^{-1}\mathbbm{q}^*_{y'},\\
&P'_{X|Y'}(x|y')=0,\quad x\in\mathcal{X}\backslash\mathcal{S},
\end{align}
where $\mathbbm{p}'_{y'}(\mathcal{S}):=(P'_{X|Y'}(x|y'))^T_{x\in\mathcal{S}}$ and 
$\mathbbm{q}^*_{y'}:=(P^*_{Y|Y'}(y|y'))_{y\in\mathcal{Y}}^T$ for $y'\in\mathcal{Y}$.
We have
\begin{align}
	\mathbbm{1}^T\mathbbm{p}'_{y'}(\mathcal{S})&=\mathbbm{1}^T\mathbb{Q}^{(c)}_{y'}(\mathcal{S})\mathbbm{p}'_{y'}(\mathcal{S})\nonumber\\
	&=\mathbbm{1}^T\mathbbm{q}^*_{y'}\nonumber\\
	&=1
\end{align}
and
\begin{align}
	\lim\limits_{\delta\rightarrow 0}\mathbbm{p}'_{y'}(\mathcal{S})&= (\mathbb{W}(\mathcal{S}))^{-1}\mathbbm{q}^{(W)}\nonumber\\
	&=\mathbbm{p}^{(W)}(\mathcal{S}),
\end{align}
where $\mathbbm{q}^{(W)}:=(P^{(W)}_Y(y))_{y\in\mathcal{Y}}^T$ and $\mathbbm{p}^{(W)}(\mathcal{S}):=(P^{(W)}_X(x))_{x\in\mathcal{S}}^T$. Therefore, when $\delta$ is sufficiently small, the constructed $P'_{X|Y'}$ is indeed a valid conditional distribution.
Note that $P'_{X|Y'}$ assigns all probability mass to $x\in\mathcal{S}$ while preserving $P^*_{Y|Y'}$. As a consequence, $P'_{XY'}:=P'_{X|Y'}P^*_{Y'}$ satisfies the constraint \eqref{eq:feedback_constraint}. We shall show that $P'_{XY'}$ 
yields a strictly larger value of $I(X;Y|Y')$ than $P^*_{XY'}$, leading to a contradiction.
Let
$\mathbbm{p}^*_{y'}(\mathcal{A}):=(P^*_{X|Y'}(x|y'))_{x\in\mathcal{A}}^T$ for $\mathcal{A}\subseteq\mathcal{X}$ and $y'\in\mathcal{Y}$.
Note that
\begin{align}
	\mathbbm{q}^*_{y'}=\mathbb{Q}^{(c)}_{y'}(\mathcal{S})\mathbbm{p}^*_{y'}(\mathcal{S})+\mathbb{Q}^{(c)}_{y'}(\mathcal{X}\backslash\mathcal{S})\mathbbm{p}^*_{y'}(\mathcal{X}\backslash\mathcal{S}),
\end{align}
which implies
\begin{align}
\mathbbm{p}^*_{y'}(\mathcal{S})=(\mathbb{Q}^{(c)}_{y'}(\mathcal{S}))^{-1}(\mathbbm{q}^*_{y'}-\mathbb{Q}^{(c)}_{y'}(\mathcal{X}\backslash\mathcal{S})\mathbbm{p}^*_{y'}(\mathcal{X}\backslash\mathcal{S})).
\end{align}
Therefore, we have
\begin{align}
\mathbbm{p}'_{y'}(\mathcal{S})-\mathbbm{p}^*_{y'}(\mathcal{S})=(\mathbb{Q}^{(c)}_{y'}(\mathcal{S}))^{-1}\mathbb{Q}^{(c)}_{y'}(\mathcal{X}\backslash\mathcal{S})\mathbbm{p}^*_{y'}(\mathcal{X}\backslash\mathcal{S}).
\end{align}
It can be verified that
\begin{align}
	\sum\limits_{x\in\mathcal{S}}(P'_{X|Y'}(x|y')-P^*_{X|Y'}(x|y'))&=\mathbbm{1}^T(\mathbbm{p}'_{y'}(\mathcal{S})-\mathbbm{p}^*_{y'}(\mathcal{S}))\nonumber\\
	&=\mathbbm{1}^T\mathbb{Q}^{(c)}_{y'}(\mathcal{S})(\mathbbm{p}'_{y'}(\mathcal{S})-\mathbbm{p}^*_{y'}(\mathcal{S}))\nonumber\\
	&=\mathbbm{1}^T\mathbb{Q}^{(c)}_{y'}(\mathcal{X}\backslash\mathcal{S})\mathbbm{p}^*_{y'}(\mathcal{X}\backslash\mathcal{S})\nonumber\\
	&=\mathbbm{1}^T\mathbbm{p}^*_{y'}(\mathcal{X}\backslash\mathcal{S})\nonumber\\
	&=\sum\limits_{x\in\mathcal{X}\backslash\mathcal{S}}P^*_{X|Y'}(x|y')\label{eq:no_absolute_value}
\end{align}
and
\begin{align}
\sum\limits_{x\in\mathcal{S}}|P'_{X|Y'}(x|y')-P^*_{X|Y'}(x|y')|
&\leq|\mathcal{S}|\|(\mathbb{Q}^{(c)}_{y'}(\mathcal{S}))^{-1}\|_{\infty}\mathbbm{1}^T\mathbb{Q}^{(c)}_{y'}(\mathcal{X}\backslash\mathcal{S})\mathbbm{p}^*_{y'}(\mathcal{X}\backslash\mathcal{S})\nonumber\\
&=|\mathcal{S}|\|(\mathbb{Q}^{(c)}_{y'}(\mathcal{S}))^{-1}\|_{\infty}\sum\limits_{x\in\mathcal{X}\backslash\mathcal{S}}P^*_{X|Y'}(x|y').\label{eq:absolute_value}
\end{align}
Define\footnote{The quantities $\overline{C}_{y'}(\mathcal{S})$, $\underline{C}_{y'}(\mathcal{S})$, and $\overline{C}_{y'}(\mathcal{X}\backslash\mathcal{S})$ are  non-negative,  since they are defined as the extrema of a set of Kullback-Leibler divergences.}
\begin{align}
	&\overline{C}_{y'}(\mathcal{S}):=\max\limits_{x\in\mathcal{S}}\sum\limits_{y\in\mathcal{Y}}Q(y|x,y')\log\frac{Q(y|x,y')}{P^*_{Y|Y'}(y|y')},\\
	&\underline{C}_{y'}(\mathcal{S}):=\min\limits_{x\in\mathcal{S}}\sum\limits_{y\in\mathcal{Y}}Q(y|x,y')\log\frac{Q(y|x,y')}{P^*_{Y|Y'}(y|y')},\\
		&\overline{C}_{y'}(\mathcal{X}\backslash\mathcal{S}):=\max\limits_{x\in\mathcal{X}\backslash\mathcal{S}}\sum\limits_{y\in\mathcal{Y}}Q(y|x,y')\log\frac{Q(y|x,y')}{P^*_{Y|Y'}(y|y')}.
\end{align}
In view of \eqref{eq:capacity} and \eqref{eq:slackness3},
\begin{align}
	&\lim\limits_{\delta\rightarrow 0}\overline{C}_{y'}(\mathcal{S})=	\lim\limits_{\delta\rightarrow 0}\underline{C}_{y'}(\mathcal{S})=C(W),\label{eq:invoke1}\\
	&\lim\limits_{\delta\rightarrow0}\overline{C}_{y'}(\mathcal{X}\backslash\mathcal{S})=C'(W):=\max\limits_{x\in\mathcal{X}\backslash\mathcal{S}}\sum\limits_{y\in\mathcal{Y}}W(y|x)\log\frac{W(y|x)}{P^{(W)}_Y(y)}<C(W).\label{eq:invoke2}
\end{align}
Observe that
\begin{align}
	&\left.I(X;Y|Y'=y')\right|_{P_{XY'}=P'_{XY'}}-\left.I(X;Y|Y'=y')\right|_{P_{XY'}=P^*_{XY'}}\nonumber\\
	&=\sum\limits_{x\in\mathcal{X}}(P'_{X|Y'}(x|y')-P^*_{X|Y'}(x|y'))\sum\limits_{y\in\mathcal{Y}}Q(y|x,y')\log\frac{Q(y|x,y')}{P^*_{Y|Y'}(y|y')}\nonumber\\
	&=\sum\limits_{x\in\mathcal{S}}(P'_{X|Y'}(x|y')-P^*_{X|Y'}(x|y'))\sum\limits_{y\in\mathcal{Y}}Q(y|x,y')\log\frac{Q(y|x,y')}{P^*_{Y|Y'}(y|y')}\nonumber\\
	&\quad-\sum\limits_{x\in\mathcal{X}\backslash\mathcal{S}}P^*_{X|Y'}(x|y')\sum\limits_{y\in\mathcal{Y}}Q(y|x,y')\log\frac{Q(y|x,y')}{P^*_{Y|Y'}(y|y')}.\label{eq:sum_bound}
\end{align}
Let $\mathcal{S}^+_{y'}:=\{x\in\mathcal{S}:P'_{X|Y'}(x|y')\geq P^*_{X|Y'}(x|y')\}$ and $\mathcal{S}^-_{y'}:=\{x\in\mathcal{S}:P'_{X|Y'}(x|y')< P^*_{X|Y'}(x|y')\}$.
We have
\begin{align}
	&\sum\limits_{x\in\mathcal{S}}(P'_{X|Y'}(x|y')-P^*_{X|Y'}(x|y'))\sum\limits_{y\in\mathcal{Y}}Q(y|x,y')\log\frac{Q(y|x,y')}{P^*_{Y|Y'}(y|y')}\nonumber\\
		&=\sum\limits_{x\in\mathcal{S}^+_{y'}}(P'_{X|Y'}(x|y')-P^*_{X|Y'}(x|y'))\sum\limits_{y\in\mathcal{Y}}Q(y|x,y')\log\frac{Q(y|x,y')}{P^*_{Y|Y'}(y|y')}\nonumber\\
	&\quad+\sum\limits_{x\in\mathcal{S}^-_{y'}}(P'_{X|Y'}(x|y')-P^*_{X|Y'}(x|y'))\sum\limits_{y\in\mathcal{Y}}Q(y|x,y')\log\frac{Q(y|x,y')}{P^*_{Y|Y'}(y|y')}\nonumber\\
	&\geq\sum\limits_{x\in\mathcal{S}^+_{y'}}(P'_{X|Y'}(x|y')-P^*_{X|Y'}(x|y'))\underline{C}_{y'}(\mathcal{S})+\sum\limits_{x\in\mathcal{S}^-_{y'}}(P'_{X|Y'}(x|y')-P^*_{X|Y'}(x|y'))\overline{C}_{y'}(\mathcal{S})\nonumber\\
	&=\sum\limits_{x\in\mathcal{S}}(P'_{X|Y'}(x|y')-P^*_{X|Y'}(x|y'))\underline{C}_{y'}(\mathcal{S})+\sum\limits_{x\in\mathcal{S}^-_{y'}}(P'_{X|Y'}(x|y')-P^*_{X|Y'}(x|y'))(\overline{C}_{y'}(\mathcal{S})-\underline{C}_{y'}(\mathcal{S}))\nonumber\\
	&\geq\sum\limits_{x\in\mathcal{S}}(P'_{X|Y'}(x|y')-P^*_{X|Y'}(x|y'))\underline{C}_{y'}(\mathcal{S})-\sum\limits_{x\in\mathcal{S}}|P'_{X|Y'}(x|y')-P^*_{X|Y'}(x|y')|(\overline{C}_{y'}(\mathcal{S})-\underline{C}_{y'}(\mathcal{S})),
\end{align}
which, together with \eqref{eq:no_absolute_value} and \eqref{eq:absolute_value}, yields
\begin{align}
&\sum\limits_{x\in\mathcal{S}}(P'_{X|Y'}(x|y')-P^*_{X|Y'}(x|y'))\sum\limits_{y\in\mathcal{Y}}Q(y|x,y')\log\frac{Q(y|x,y')}{P^*_{Y|Y'}(y|y')}\nonumber\\
&\geq (\underline{C}_{y'}(\mathcal{S})-|\mathcal{S}|\|(\mathbb{Q}^{c}_{y'}(\mathcal{S}))^{-1}\|_{\infty}(\overline{C}_{y'}(\mathcal{S})-\underline{C}_{y'}(\mathcal{S})))\sum\limits_{x\in\mathcal{X}\backslash\mathcal{S}}P^*_{X|Y'}(x|y').\label{eq:bound1}
\end{align}
Moreover,
\begin{align}
\sum\limits_{x\in\mathcal{X}\backslash\mathcal{S}}P^*_{X|Y'}(x|y')\sum\limits_{y\in\mathcal{Y}}Q(y|x,y')\log\frac{Q(y|x,y')}{P^*_{Y|Y'}(y|y')}&\leq\sum\limits_{x\in\mathcal{X}\backslash\mathcal{S}}P^*_{X|Y'}(x|y')\overline{C}_{y'}(\mathcal{X}\backslash\mathcal{S})\nonumber\\
&=\overline{C}_{y'}(\mathcal{X}\backslash\mathcal{S})\sum\limits_{x\in\mathcal{X}\backslash\mathcal{S}}P^*_{X|Y'}(x|y').\label{eq:bound2}
\end{align}
Substituting \eqref{eq:bound1} and \eqref{eq:bound2} into \eqref{eq:sum_bound} gives
\begin{align}
	&\left.I(X;Y|Y'=y')\right|_{P_{XY'}=P'_{XY'}}-\left.I(X;Y|Y'=y')\right|_{P_{XY'}=P^*_{XY'}}\nonumber\\
	&\geq(\underline{C}_{y'}(\mathcal{S})-\overline{C}_{y'}(\mathcal{X}\backslash\mathcal{S})-|\mathcal{S}|\|(\mathbb{Q}^{c}_{y'}(\mathcal{S}))^{-1}\|_{\infty}(\overline{C}_{y'}(\mathcal{S})-\underline{C}_{y'}(\mathcal{S})))\sum\limits_{x\in\mathcal{X}\backslash\mathcal{S}}P^*_{X|Y'}(x|y').
\end{align}
Since $\mathbb{Q}^{(c)}_{y'}(\mathcal{S})\rightarrow\mathbb{W}(\mathcal{S})$ as $\delta\rightarrow 0$, invoking \eqref{eq:invoke1} and \eqref{eq:invoke2} shows
\begin{align}
	\lim\limits_{\delta\rightarrow 0}\underline{C}_{y'}(\mathcal{S})-\overline{C}_{y'}(\mathcal{X}\backslash\mathcal{S})-|\mathcal{S}|\|(\mathbb{Q}^{c}_{y'}(\mathcal{S}))^{-1}\|_{\infty}(\overline{C}_{y'}(\mathcal{S})-\underline{C}_{y'}(\mathcal{S}))=C(W)-C'(W)>0.
\end{align}
Therefore, when $\delta$ is sufficiently small,
\begin{align}
	\left.I(X;Y|Y'=y')\right|_{P_{XY'}=P'_{XY'}}-\left.I(X;Y|Y'=y')\right|_{P_{XY}=P^*_{XY'}}\geq\frac{1}{2}(C(W)-C'(W))\sum\limits_{x\in\mathcal{X}\backslash\mathcal{S}}P^*_{X|Y'}(x|y')
\end{align}
for all $y'\in\mathcal{Y}$, and consequently
\begin{align}
&\left.I(X;Y|Y')\right|_{P_{XY'}=P'_{XY'}}-\left.I(X;Y|Y')\right|_{P_{XY'}=P^*_{XY'}}\nonumber\\
&=\sum\limits_{y'\in\mathcal{Y}}P^*_{Y'}(y')(\left.I(X;Y|Y'=y')\right|_{P_{XY'}=P'_{XY'}}-\left.I(X;Y|Y'=y')\right|_{P_{XY'}=P^*_{XY'}})\nonumber\\
&\geq\frac{1}{2}(C(W)-C'(W))\sum\limits_{y'\in\mathcal{Y}}P^*_{Y'}(y)\sum\limits_{x\in\mathcal{X}\backslash\mathcal{S}}P^*_{X|Y'}(x|y')\nonumber\\
&=\frac{1}{2}(C(W)-C'(W))\sum\limits_{x\in\mathcal{X}\backslash\mathcal{S}}P^*_{X}(x)\nonumber\\
&>0,
\end{align}
contradicting the fact that $P^*_{XY'}$ maximizes $I(X;Y|Y')$ subject to the constraint \eqref{eq:feedback_constraint}. This completes the proof of Lemma \ref{lem:refined_version}.
	\end{IEEEproof}
	
Lemma \ref{lem:refined_version} shows that, for any 
$W$-centered 
$\delta$-approximately memoryless surjective POST channel, only the subset 
$\mathcal{S}$ of the input alphabet $\mathcal{X}$ needs to be used in order to achieve the feedback capacity when 
$\delta$ is sufficiently small. Because this subset has cardinality 
$|\mathcal{S}|=|\mathcal{Y}|$, the problem effectively reduces to one with matching input and output alphabet sizes. The latter case has already been fully resolved, thereby completing the argument.

 %Moreover, $P'_{X|Y'}(x|y')$  coincides $P^*_{X|Y'}(x|y')$ for all $x\in\mathcal{X}$ if and only if $\sum_{x\in\mathcal{S}}P^*_{X|Y'}(x|y')=1$.
%Note that $P'_{XY'}:=P'_{X|Y'}P^*_{Y'}$ satisfies the constraint \eqref{eq:feedback_constraint} because replacing $P^*_{X|Y'}$ with $P'_{X|Y'}$ preserves $P*_{Y|Y'}$.
%For $y'$ such that $\sum_{x\in\mathcal{S}}P^*_{X|Y'}(x|y')=1$, we have
%\begin{align}
%	\left.I(X;Y|Y'=y')\right|_{P'}=\left.I(X;Y|Y'=y')\right|_{P^*}.
%\end{align}
%For $y'$ such that $\sum_{x\in\mathcal{S}}P^*_{X|Y'}(x|y')=1$, we have

%\begin{align}
%	\left.I(X;Y|Y'=y')\right|_{P_{XY'}=P'_{XY'}}=\left.\sum\limits_{y'\in\mathcal{Y}}P^*_{Y'}(y')I(X;Y|Y'=y')\right|_{P_{X|Y'}=P'_{X|Y'}}
%\end{align}

%$P^{(Q)}_{XY'}$ converges implies $P^{(Q)}_{YY'}$ converges, which further implies $P^{(Q)}_{Y|Y'}$ converges. The full rank condition implies that for each $y'$, we can find $p_{X|y'}$ to preserve $P_{Y|y'}$ while increasing $I(X;Y|Y'=y')$ as compared to an assignment that assigns some probability outside $\mathcal{A}$.

%This lemma implies uniqueness. 

\section{Discussion}\label{sec:discussion}

Although the proof of Theorem \ref{thm:main} is technically involved, the underlying ideas are conceptually simple. We briefly summarize them here to clarify the logic of the argument and to highlight the role played by the surjectivity condition.

First consider the case $|\mathcal{X}|=|\mathcal{Y}|$. When the $n$-th fold channel transition matrix $\mathbb{Q}^{(n)}_{y_0}$ is full rank (i.e., of rank $|\mathcal{X}|^n$),   the output probability vector $\mathbbm{q}^{(n)}_{y_0}$ induced by the optimal feedback strategy---in fact, any output probability vector of dimension $|\mathcal{Y}|^n=|\mathcal{X}|^n$---is necessarily contained in the linear span of the column vectors $\mathbbm{v}^{(n)}_{y_0,0}, \mathbbm{v}^{(n)}_{y_0,1}, \ldots,\mathbbm{v}^{(n)}_{y_0,|\mathcal{X}|^n-1}$ of $\mathbb{Q}^{(n)}_{y_0}$. Consequently, there exist coefficients $\beta_0,\beta_1,\ldots,\beta_{|\mathcal{X}|^n-1}$ satisfying
\begin{align} \sum_{i=0}^{|\mathcal{X}|^n-1}\beta_i=1 \label{eq:linear_span}
\end{align}	
such that
\begin{align}
	\sum\limits_{i=0}^{|\mathcal{X}|^n-1}\beta_i\mathbbm{v}^{(n)}_{y_0,i}=\mathbbm{q}^{(n)}_{y_0}.\label{eq:linear_space}
\end{align}
We shall refer to the space defined by \eqref{eq:linear_space} subject to the normalization constraint \eqref{eq:linear_span} as the affine span of the column vectors of $\mathbb{Q}^{(n)}_{y_0}$.
Note that inclusion in the affine span does not, by itself, ensure that
 $\mathbbm{q}^{(n)}_{y_0}$ is realizable by an input distribution without feedback. For this, $\mathbbm{q}^{(n)}_{y_0}$ must lie in the convex hull of these column vectors, which requires each coefficient to be nonnegative:
\begin{align}
	\beta_i\geq 0\label{eq:convexhull}
\end{align}
for $i=0,1,\ldots,|\mathcal{X}|^n-1$. Note that \eqref{eq:linear_span} and \eqref{eq:convexhull} correspond to \eqref{eq:normalization} and \eqref{eq:convex_hull}, respectively.

\begin{figure}[htbp]
	\centerline{\includegraphics[width=15cm]{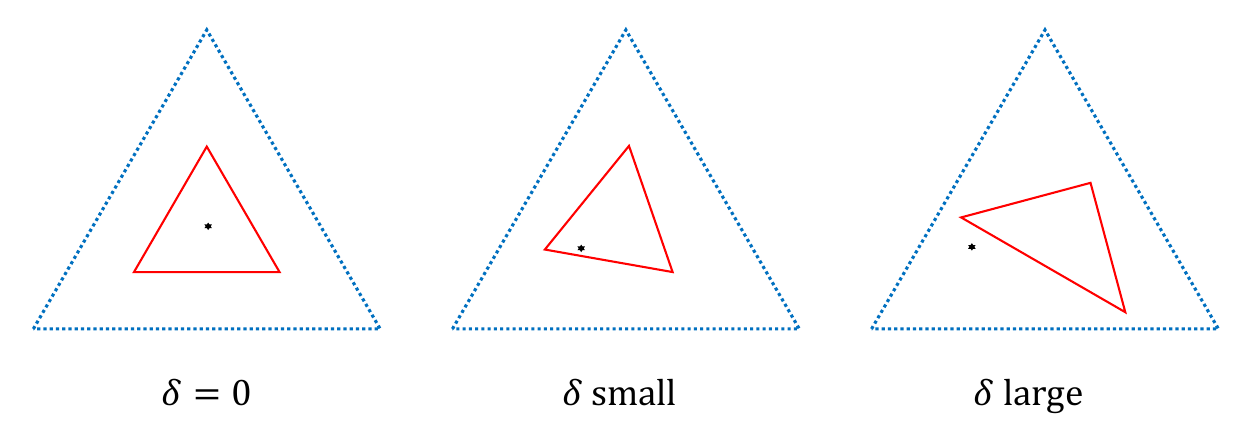}} \caption{In the provided visualization, the dashed triangle, the solid triangle, and the asterisk represent the affine span of the columns of $\mathbb{Q}^{(n)}_{y_0}$, their convex hull, and the output probability vector 
		$\mathbbm{q}^{(n)}_{y_0}$
		induced by the optimal feedback strategy, respectively. Under the memoryless setting where $\delta=0$, the vector $\mathbbm{q}^{(n)}_{y_0}$
		is guaranteed to reside within the convex hull.
		When a small perturbation $\delta$ is introduced---representing an approximately memoryless regime---neither $\mathbbm{q}^{(n)}_{y_0}$ nor the convex hull undergoes significant deformation,
		maintaining their containment relationship. However, as the magnitude of the perturbation $\delta$ increases, the divergence of their relative trajectories may eventually cause   $\mathbbm{q}^{(n)}_{y_0}$ to fall outside the boundaries of the convex hull.}
	\label{fig1} 
\end{figure}

To gain insight into when \eqref{eq:convexhull} holds, it is helpful to consider the degenerate case in which the single-use channel transition matrices $\mathbb{Q}^{(c)}_0,\mathbb{Q}^{(c)}_1,\ldots,\mathbb{Q}^{(c)}_{|\mathcal{Y}|-1}$ all coincide with $\mathbb{W}$. In this setting, the full-rankness part of the simplified surjectivity condition guarantees that the 
$n$-fold channel transition matrix $\mathbb{Q}^{(n)}_{y_0}$ has rank $|\mathcal{X}|^n$. Moreover, because the channel is memoryless, the optimal feedback strategy reduces to a non-feedback strategy (i.e., using i.i.d. inputs), the corresponding output probability vector $\mathbbm{q}^{(n)}_{y_0}$ indeed lies in the convex hull of the columns of $\mathbb{Q}^{(n)}_{y_0}$. The positivity  part of the simplified surjectivity condition further ensures that $\mathbbm{q}^{(n)}_{y_0}$ lies in the strict interior of this convex hull. Consequently, under small perturbations, neither $\mathbbm{q}^{(n)}_{y_0}$ nor the convex hull\footnote{In particular,  the full rankness part of the surjectivity condition ensures that the dimension of the convex hull remains unchanged under sufficiently small perturbations.} changes substantially, and thus the former should remain contained within the latter. This provides the key intuition for why, in the case of approximately memoryless POST channels, the output process induced by the optimal feedback strategy can still be simulated in a non-feedback manner, thereby implying that feedback does not increase capacity.
See Fig. \ref{fig1} for an illustration.

The same intuition extends to the case 
$|\mathcal{X}|>|\mathcal{Y}|$. In this regime, the input space has higher dimension than the output space, providing additional flexibility for shaping the output process without feedback\footnote{Specifically, the dimension of the convex hull of the columns of $\mathbb{Q}^{(n)}_{y_0}$ remains $|\mathcal{Y}|^n-1$, but it may have up to $|\mathcal{X}|^n$ extreme points, as opposed to only $|\mathcal{Y}|^n$ extreme points when $|\mathcal{X}|=|\mathcal{Y}|$}. For memoryless channels,  one never needs more input symbols than output symbols to achieve capacity. For POST channels, the situation is subtler: the optimal strategy may, in principle, use different subsets of input symbols under different channel states, and this added flexibility introduces analytical challenges. The strict-complementary-slackness part of the surjectivity condition is designed precisely to rule out such complications. It guarantees the existence of a fixed subset of input symbols, of cardinality 
$|\mathcal{Y}|$, that is uniformly strictly better than the remaining symbols. As a result, under sufficiently small perturbations, this same subset continues to suffice across all channel states, allowing the problem to be smoothly reduced to the equal–alphabet–size case.

With this understanding, it becomes clear why the case $|\mathcal{X}|<|\mathcal{Y}|$ is fundamentally different. In this regime, the dimension of the affine span of the columns of $\mathbb{Q}^{(n)}_{y_0}$  is at most $|\mathcal{X}|^n-1$, which is strictly smaller than the dimension of the ambient space in which the output probability vector $\mathbbm{q}^{(n)}_{y_0}$ resides, namely $|\mathcal{Y}|^n-1$. Although in the memoryless case the corresponding $\mathbbm{q}^{(n)}_{y_0}$ still lies in the convex hull of the columns of $\mathbb{Q}^{(n)}_{y_0}$, the presence of additional dimensions in the ambient output space means that even a small perturbation may cause $\mathbbm{q}^{(n)}_{y_0}$ to move outside  the linear span of these columns, let alone their affine span or convex hull. For essentially the same reason, extending the argument to the case $|\mathcal{X}|=|\mathcal{Y}|$ with rank-deficient channel transition matrices rank deficiency is also nontrivial. See Fig. \ref{fig2} for an illustration.

\begin{figure}[htbp]
	\centerline{\includegraphics[width=10cm]{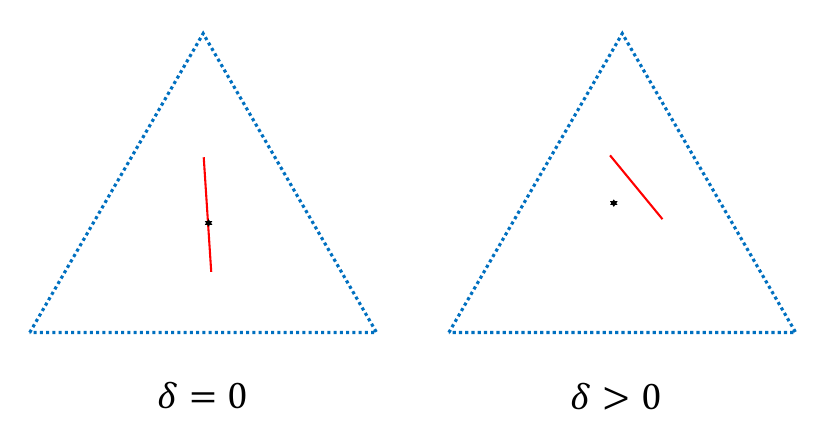}} \caption{In this visualization,  the dashed triangle, the solid line segment, and the asterisk represent the ambient output space, the convex hull of the columns of $\mathbb{Q}^{(n)}_{y_0}$, and the output probability vector 
		$\mathbbm{q}^{(n)}_{y_0}$ 	induced by the optimal feedback strategy, respectively. Under the memoryless setting ($\delta=0$), $\mathbbm{q}^{(n)}_{y_0}$ is guaranteed to reside within the convex hull. However, as the perturbation $\delta$ deviates even slightly from zero, 
		the availability of an extra dimension in the  ambient output space  allows $\mathbbm{q}^{(n)}_{y_0}$ to move beyond  the  linear span of the columns of $\mathbb{Q}^{(n)}_{y_0}$, breaking the containment relationship.}
	\label{fig2} 
\end{figure}

From this perspective, feedback strategies possess additional degrees of freedom compared with non-feedback strategies:
\begin{enumerate}
	\item The output probability vector $\mathbbm{q}^{(n)}_{y_0}$ can move into regions of the affine span  of the columns of $\mathbb{Q}^{(n)}_{y_0}$ that lie outside their convex hull.

	\item The output probability vector $\mathbbm{q}^{(n)}_{y_0}$ is not confined to this affine span and may instead enter the additional dimensions available in the ambient output space.
	\end{enumerate}
	When the input dominates (i.e., $\mathbb{Q}^{(n)}_{y_0}$ has full row rank), only the first degree of freedom is available; When the output dominates (i.e.,  $\mathbb{Q}^{(n)}_{y_0}$ is row rank-deficient), both become available. Whether feedback increases capacity or not can thus be interpreted as a manifestation of whether exploiting these extra degrees of freedom provides any advantage.

To better understand the second degree of freedom and to see why the surjectivity condition is not superfluous, consider the following illustrative examples.
\begin{itemize}
	\item Example 1: $|\mathcal{X}|=2$ and $|\mathcal{Y}|=3$. Let the transition matrix $\mathbb{W}$ of the  memoryless reference channel be given by $\mathbb{W}:=(\mathbbm{w}_0,\mathbbm{w}_1)$,
	where
	\begin{align}
		\mathbbm{w}_0:=\left(\begin{matrix}
			\frac{2}{3}\\
			\frac{1}{6}\\
			\frac{1}{6}
		\end{matrix}\right),\quad \mathbbm{w}_1:=\left(\begin{matrix}
		\frac{1}{5}\\
		\frac{3}{5}\\
		\frac{1}{5}
		\end{matrix}\right).
	\end{align}
	Since $|\mathcal{X}|<|\mathcal{Y}|$, the surjectivity condition is violated.
	Define $\mathbb{U}_{y'}:=(\mathbbm{u}_{y',0},\mathbbm{u}_{y',1})$ for $y'=0,1,2$, with
	\begin{align}
		&\mathbbm{u}_{0,0}:=-\frac{1}{2}\mathbbm{w}_0+\frac{1}{2}\mathbbm{w}_1,\quad\mathbbm{u}_{0,1}:=\frac{1}{3}\mathbbm{w}_0-\frac{1}{3}\mathbbm{w}_1,\label{eq:u1}\\
		&\mathbbm{u}_{1,0}:=-\frac{2}{5}\mathbbm{w}_0+\frac{2}{5}\mathbbm{w}_1,\quad\mathbbm{u}_{1,1}:=\frac{1}{2}\mathbbm{w}_0-\frac{1}{2}\mathbbm{w}_1,\label{eq:u2}\\
		&\mathbbm{u}_{2,0}:=-\frac{2}{3}\mathbbm{w}_0+\frac{2}{3}\mathbbm{w}_1,\quad\mathbbm{u}_{2,1}:=\frac{3}{5}\mathbbm{w}_0-\frac{3}{5}\mathbbm{w}_1,\label{eq:u3}
	\end{align}
	and set the channel transition matrices of the POST channel 
	\begin{align}
		\mathbb{Q}^{(c)}_{y'}:=\mathbb{W}+\epsilon\mathbb{U}_{y'}\label{eq:POST}
	\end{align}
	for $y'=0,1,2$. Following \eqref{eq:recursionmatrix1}, we  recursively construct the $n$-th fold channel transition matrix $\mathbb{Q}^{(n)}_{y_0}$.
	Moreover, the Markov  kernel $P^*_{Y|Y'}$ can be obtained numerically by solving the optimization problem in \eqref{eq:feedback_capacity}--\eqref{eq:feedback_constraint}. Using this kernel, the output probability vector
 $\mathbbm{q}^{(n)}_{y_0}$ induced by the optimal feedback strategy can be constructed via \eqref{eq:recursionmatrix2}.
	
	Now consider  the case $n=2$. Observe that $\mathbb{Q}^{(2)}_{y_0}$ is a $9\times 4$ matrix of full column rank. Define
	\begin{align}
		\mathbbm{p}^{(2)}_{y_0}:=\mathbb{G}^{(2)}_{y_0}\mathbbm{q}^{(2)}_{y_0},
	\end{align}
	where
	 $\mathbb{G}^{(2)}_{y_0}$ denotes the Moore–Penrose pseudoinverse of $\mathbb{Q}^{(2)}_{y_0}$. Then
	 \begin{align} \hat{\mathbbm{q}}^{(2)}_{y_0}:=\mathbb{Q}^{(2)}_{y_0}\mathbbm{p}^{(2)}_{y_0}
	 	\end{align}
	 	 provides the least-squares approximation of $\mathbbm{q}^{(2)}_{y_0}$ using the columns of $\mathbb{Q}^{(2)}_{y_0}$. If $\|\hat{\mathbbm{q}}^{(2)}_{y_0}-\mathbbm{q}^{(2)}_{y_0}\|_1>0$, then $\mathbbm{q}^{(2)}_{y_0}$ does not lie in the linear span of the columns of $\mathbb{Q}^{(2)}_{y_0}$, let alone their affine span or convex hull. Define
	 	 \begin{align}
	 	 	D:=\sum\limits_{y_0=0}^2\|\hat{\mathbbm{q}}^{(2)}_{y_0}-\mathbbm{q}^{(2)}_{y_0}\|_1.\label{eq:D}
	 	 \end{align}
	 	 We plot $D$ vs. $\epsilon$ in Fig. \ref{example1}. As can be seen,
	 	  $D$ becomes strictly positive as soon as $\epsilon$ is non-zero. This indicates that even a small perturbation can cause $\mathbbm{q}^{(2)}_{y_0}$ to leave the columnn space of $\mathbb{Q}^{(2)}_{y_0}$, rendering it unrealizable by any non-feedback strategy.

\begin{figure}[htbp]
	\centerline{\includegraphics[width=12cm]{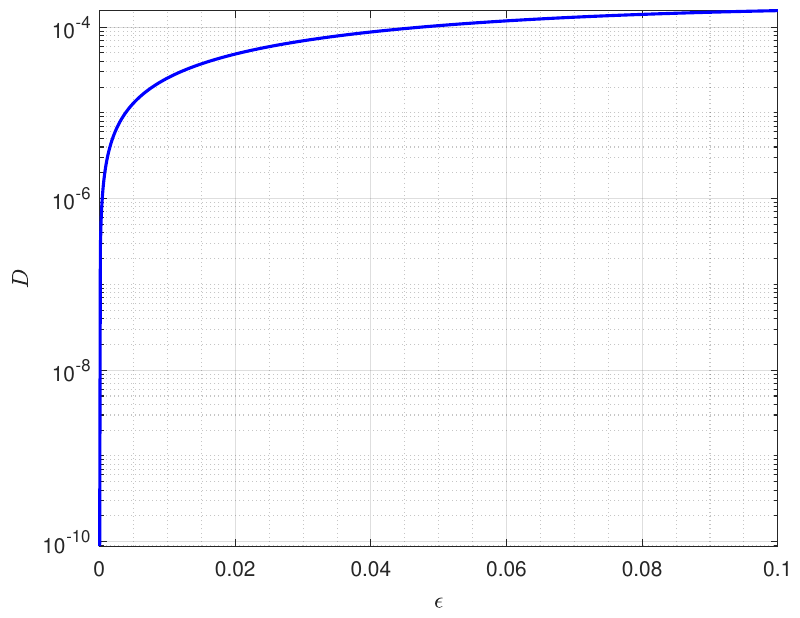}} \caption{Plot of $D$ vs. $\epsilon$. The positivity of $D$ shows that, for at least some initial state $y_0$, the output probability vector $q^{(2)}_{y_0}$ no longer resides in the column space of the channel transition matrix $\mathbb{Q}^{(2)}_{y_0}$, but instead occupies directions in the ambient output space that are inaccessible to non-feedback strategies.}
	\label{example1} 
\end{figure}

\item Example 2: $|\mathcal{X}|=|\mathcal{Y}|=3$. Let the transition matrix $\mathbb{W}$ of the memoryless reference channel be given by $\mathbb{W}:=(\mathbbm{w}_0,\mathbbm{w}_1,\mathbbm{w}_2)$,
where
\begin{align}
	\mathbbm{w}_0:=\left(\begin{matrix}
		\frac{1}{2}\\
		\frac{1}{3}\\
		\frac{1}{6}
	\end{matrix}\right),\quad \mathbbm{w}_1:=\left(\begin{matrix}
		\frac{1}{4}\\
		\frac{1}{2}\\
		\frac{1}{4}
	\end{matrix}\right),\quad \mathbbm{w}_2:=\frac{2}{3}\mathbbm{w}_0+\frac{1}{3}\mathbbm{w}_1.
\end{align}
Since $\mathbb{W}$ is rank-deficient, the surjectivity condition is violated. For $y'=0,1,2$, let $\mathbb{U}_{y'}:=(\mathbbm{u}_{y',0},\mathbbm{u}_{y',1},\mathbbm{u}_{y',2})$, where $\mathbbm{u}_{y',0}$ and $\mathbbm{u}_{y',1}$ are defined in the same manner as in \eqref{eq:u1}--\eqref{eq:u3}, and
\begin{align}
\mathbbm{u}_{y',2}:=\frac{2}{3}\mathbbm{u}_{y',0}+\frac{1}{3}\mathbbm{u}_{y',1}.
\end{align}
The channel transition matrices $\mathbbm{Q}^{(c)}_0,\mathbbm{Q}^{(c)}_1,\mathbbm{Q}^{(c)}_2$ of the POST channel are then defined according to \eqref{eq:POST}.
As in the first example, we can recursively construct $\mathbb{Q}^{(n)}_{y_0}$ and $\mathbbm{q}^{(n)}_{y_0}$.

Now consider the case $n=2$. In the current example, $\mathbb{Q}^{(2)}_{y_0}$ is a $9\times 9$ matrix of rank $4$. We define $D$ in the same way as in \eqref{eq:D} and plot it as a function of $\epsilon$  in Fig. \ref{example2}. Again, $D$ becomes strictly positive once $\epsilon$ deviates from zero. This shows that column-rank deficiency in the equal-alphabet setting can lead to the same phenomenon observed when the output alphabet is larger than the input alphabet.

\begin{figure}[htbp]
	\centerline{\includegraphics[width=12cm]{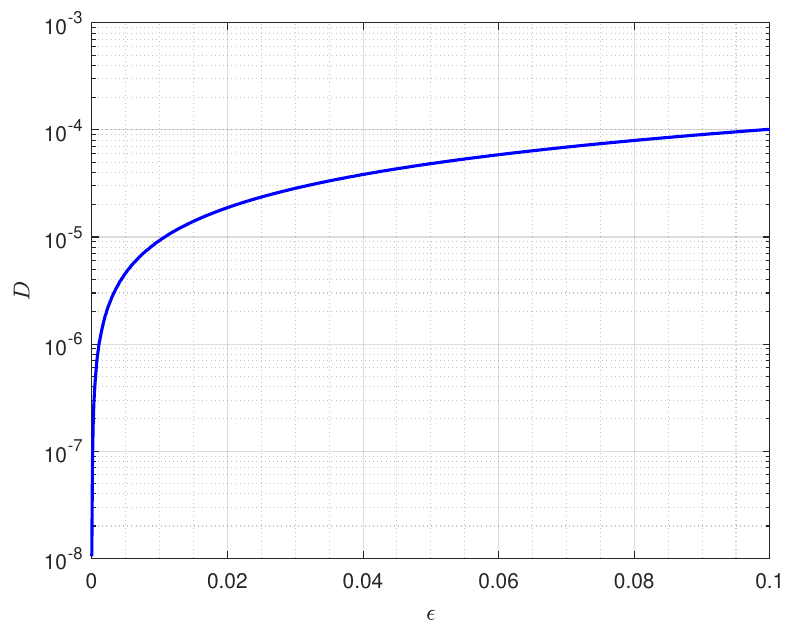}} \caption{Plot of $D$ vs. $\epsilon$. The positivity of 
	$D$ shows that rank deficiency in the equal-alphabet setting can render certain dimensions of the ambient output space inaccessible to non-feedback strategies, whereas the optimal feedback strategy actively exploits these dimensions even when the channel exhibits only slight memory.}
	\label{example2} 
\end{figure}

\end{itemize}

The bulk of this work relies on the fact that the realizability of the Markov process induced by the optimal feedback strategy through a POST channel via a non-feedback strategy serves as a sufficient condition for the non-feedback capacity to coincide with the feedback capacity. We conclude this section by showing that this condition is also necessary.
Consequently,  the inability to simulate this Markov process---as seen in the preceding examples---should be viewed as fundamental evidence that feedback strictly improves capacity, rather than a limitation of our methodology.

% Therefore, the non-existence of a non-feedback strategy to realize the output probability vector induced by the optimal feedback strategy revealed in the previous two examples should  not be interpreted as a limitation of our proof technique, but rather as the evidence that feedback strictly improves capacity.

\begin{theorem}\label{thm:necessary_sufficient}
	For any POST channel $Q$, 
	\begin{align}
		C^*(Q)=C^*_f(Q)
	\end{align}
	if and only if, for every positive integer $n$, there exist a joint input-initial-state distribution $P_{X^nY_0}$ and a maximizer $P^*_{XY'}$ of the optimization problem in \eqref{eq:feedback_capacity}--\eqref{eq:feedback_constraint} 
	such that their corresponding distributions $P_{X^nY^nY_0}$ and $P^*_{XYY'}$, induced through the POST channel $Q$, satisfy
	\begin{align}
		P_{Y^nY_0}(y^n,y_0)=P^*_{Y'}(y_0)\prod\limits_{t=1}^nP^*_{Y|Y'}(y_t|y_{t-1})\label{eq:first_condition}
	\end{align}
	and 
	\begin{align}
		P_{X_tY_{t-1}}=P^*_{XY'}\label{eq:second_condition}
	\end{align}
	for $t=1,2,\ldots,n$, where $C^*(Q)$ and $C^*_f(Q)$ are defined in \eqref{eq:nonfeedback_capacity} and \eqref{eq:feedback_capacity}, respectively.
\end{theorem}

%\begin{theorem}\label{thm:necessary_sufficient}
%	Given a memoryless surjective channel $W$, when $\delta$ is sufficiently small,
%	\begin{align}
%		C(Q)=C_f(Q)
%	\end{align}
%	for a $W$-centered $\delta$-approximately memoryless POST channel $Q$ if and only if, for every positive integer $n$, there exists an input distribution $P^{(n)}_{X^n|Y_0}$ satisfying \eqref{eq:Markov}.
%	\end{theorem}
	\begin{IEEEproof}
	See Appendix \ref{app:necessary_sufficient}.
		\end{IEEEproof}
	
Theorem \ref{thm:necessary_sufficient} does not require the POST channel $Q$ to be indecomposable\footnote{In particular,  the limit  defining $C^*(Q)$ in  \eqref{eq:nonfeedback_capacity} always exists, regardless of the indecomposability of the POST channel
 $Q$  \cite[Theorem 4.6.1, p. 100]{Gallager68}.} and connected; these conditions are needed only for $C^*(Q)$ and $C^*_f(Q)$ to acquire their respective operational meanings via \eqref{eq:nonfeedback_capacity} and \eqref{eq:feedback_capacity}. Moreover, condition \eqref{eq:second_condition} is implied by condition \eqref{eq:first_condition} and is thus redundant if all channel transition matrices $\mathbb{Q}^{(c)}_0,\mathbb{Q}^{(c)}_1,\ldots,\mathbb{Q}^{(c)}_{|\mathcal{Y}|-1}$ have full column rank.

%limit exists even when $Q$ is decomposable. 

%\section{Gaussian Case}\label{sec:Gaussian}

%mentioned entropy constrained scalar quantizer for Gaussian source, it serves as an upper bound.

\section{Conclusion}\label{sec:conclusion}

We have shown that, for  approximately memoryless surjective POST channels, the feedback capacity coincides with the non-feedback capacity. The surjectivity condition imposed in our analysis serves primarily to streamline the proof, and  can be relaxed to some extent. Nevertheless, it remains unclear what the most general condition is under which feedback provides no capacity gain for POST channels whose input alphabet size is no smaller than the output alphabet size. In contrast, extending the result to settings where the output alphabet size is larger than the input alphabet size appears fundamentally more delicate. The main technical obstacle is to identify conditions ensuring that the distribution of the optimal output process induced under feedback lies within the convex hull of the columns of the non-feedback channel transition matrix, a requirement that does not generally hold when the output  dominates.

%Whether these conditions can be removed entirely, thereby establishing that feedback does not increase the capacity of all approximately memoryless POST channels, .

In our formulation, the transition matrices of different states of an approximately memoryless POST channel are modeled as perturbations of the transition matrix of a memoryless channel. It is therefore of significant interest to determine an effective perturbation bound under which our main result holds. While such a bound could in principle be derived by carefully tracing the inequalities throughout our proof and by strengthening Lemmas \ref{lem:uniform_convergence} and \ref{lem:refined_version} with explicit  convergence-rate estimates, the resulting expression would likely be cumbersome and offer limited conceptual insight. Developing a compact, interpretable, and operationally meaningful perturbation bound would be a valuable direction for future work.

The notion of approximately memoryless POST channels was introduced mainly for the convenience of stating our main theorem. In fact, the binary POST$(a,b)$ channel example in \cite{PAW14} demonstrates that the transition matrices need not be close at all for feedback to offer no advantage. This suggests the potential for more general formulations that extend well beyond the approximately memoryless setting. Pushing in this direction will likely require techniques that go beyond perturbation analysis. Moreover, although our focus has been on POST channels, similar phenomena may arise in other classes of channel models as well, and identifying such models would deepen our understanding of when feedback is or is not beneficial.

Finally, the core of our argument shows that, under suitable conditions, a closed-loop process can be effectively simulated by an open-loop process. This observation may have implications beyond information theory. Investigating analogous phenomena within the broader frameworks of partially observable Markov decision processes and reinforcement learning represents a promising and potentially impactful research direction.


% or


\appendices
\section{Proof of Lemma \ref{lem:indecomposable_connected}}\label{app:indecomposable_connected}

A finite square matrix $\mathbb{Q}$ is called column-stochastic if all its entries are non-negative and each column sums to 
$1$. Furthermore, a column-stochastic matrix $\mathbb{Q}$ is said to be scrambling if 
\begin{align}
	\lambda(\mathbb{Q}):=1-\min\limits_{i,j}\sum\limits_{k}\min\{\mathbb{Q}(k,i),\mathbb{Q}(k,j)\}<1.
\end{align}
It is known \cite{Hajnal58} (see also \cite[Theorem 2.3]{CPW10}) that if all state transition matrices $\mathbb{Q}^{(s)}_0,\mathbb{Q}^{(s)}_0,\ldots,\mathbb{Q}^{(s)}_{|\mathcal{X}|-1}$ are scrambling, then $Q$ is an indecomposable POST channel. For a $W$-centered $\delta$-approximately memoryless POST channel $Q$, 
\begin{align}
	\lambda(\mathbb{Q}^{(s)}_x)&=1-\min\limits_{y',y''\in\mathcal{Y}}\sum\limits_{y\in\mathcal{Y}}\min\{Q(y|x,y'),Q(y|x,y'')\}\nonumber\\
	&\leq1-\sum\limits_{y\in\mathcal{Y}}\max\{W(y|x)-\delta,0\}.
\end{align}
Thus, whenever \eqref{eq:cond_indecomposable} holds, we have $\lambda(\mathbb{Q}^{(s)}_x)<1$ for all $x\in\mathcal{X}$, and consequently $Q$ is indecomposable.

Note that 
\begin{align}
	\max\limits_{x\in\mathcal{X}}\mathbb{Q}^{(s)}_x(y,y')\geq \max\limits_{x\in\mathcal{X}}W(y|x)-\delta.
\end{align}
Thus, whenever \eqref{eq:cond_connected} holds, for any $y,y'\in\mathcal{Y}$, there exists some $x\in\mathcal{X}$ such that $\mathbb{Q}^{(s)}_x(y,y')>0$, and consequently $Q$ is connected.


%\section{Proof of Lemma \ref{lem:limit}}\label{app:limit}

\section{Proof of Lemma \ref{lem:uniform_convergence}}\label{app:uniform_convergence}

Suppose, to the contrary, that $P^*_{XY'}$ does not converge uniformly to $P^{(W)}_{X}P^{(W)}_Y$. Then there exists a sequence of $W$-centered $\delta_k$-approximately memoryless POST channels $Q^{(k)}$, with $\delta_k\rightarrow 0$ as $k\rightarrow\infty$, and  corresponding maximizers $P^*_{XY'}(Q^{(k)})$ such that
\begin{align} P^*_{XY'}(Q^{(k)})\rightarrow P'_{XY'}\neq P^{(W)}_XP^{(W)}_Y
\end{align}	
as $k\rightarrow\infty$. Note that $P'_{XY'}$ satisfies the constraint \eqref{eq:feedback_constraint} when $Q$ degenerates to $W$.
Define the Markov kernel
\begin{align}
	P^{(k)}_{Y|Y'}(y|y'):=\sum\limits_{x\in\mathcal{X}}P^{(W)}_X(x)Q^{(k)}(y|x,y').
\end{align}
Because $P^{(W)}_X(x)>0$ for all $x\in\mathcal{X}$, and 
all channel transition matrices associated with $Q^{(k)}$ have full rank whenever $\delta_k$ satisfies \eqref{eq:full_rank}, we have $P^{(k)}_{Y|Y'}(y|y')>0$ for all $y,y'\in\mathcal{Y}$ and all sufficiently large $k$. Thus the resulting Markov chain is irreducible and admits a unique stationary distribution $P^{(k)}_{Y'}$ \cite[Theorems 5.3.3, 5.5.9, and 5.5.12]{Durrett2019}. Set
\begin{align} P^{(k)}_{XY'}:=P^{(W)}_XP^{(k)}_{Y'}.
\end{align}
A direct verification shows that
$P^{(k)}_{XY'}$ satisfies the constraint \eqref{eq:feedback_constraint} with respect to $Q^{(k)}$. Hence,
\begin{align}
	\left.I(X;Y|Y')\right|_{P_{XY'}=P^*_{XY'}(Q^{(k)}),Q=Q^{(k)}}\geq \left.I(X;Y|Y')\right|_{P_{XY'}=P^{(k)}_{XY'}(Q^{(k)}),Q=Q^{(k)}}.
\end{align}
Since 
\begin{align}
	&\lim\limits_{k\rightarrow\infty}\left.I(X;Y|Y')\right|_{P_{XY'}=P^*_{XY'}(Q^{(k)}),Q=Q^{(k)}}=\left.I(X;Y|Y')\right|_{P_{XY'}=P'_{XY'},Q=W},\\
	&\lim\limits_{k\rightarrow\infty}\left.I(X;Y|Y')\right|_{P_{XY'}=P^{(k)}_{XY'}(Q^{(k)}),Q=Q^{(k)}}=\left.I(X;Y|Y')\right|_{P_{XY'}=P^{(W)}_XP^{(W)}_Y,Q=W},
\end{align}
it follows that 
\begin{align}
	\left.I(X;Y|Y')\right|_{P_{XY'}=P'_{XY'},Q=W}\geq \left.I(X;Y|Y')\right|_{P_{XY'}=P^{(W)}_XP^{(W)}_Y,Q=W}.
\end{align}
Therefore, $P'_{XY'}$ is also a maximizer of the optimization problem in \eqref{eq:feedback_capacity}--\eqref{eq:feedback_constraint} when  $Q$ degenerates to $W$. This contradicts Lemma \ref{lem:limit}, thereby establishing the desired result.

It remains to prove the uniqueness of the maximizer $P^*_{XY'}$. 
To this end, we investigate the Hessian of $I(X;Y|Y')$ with respect to $P_{XY'}$.
First write
\begin{align}
	I(X;Y|Y')=H(Y|Y')-H(Y|X,Y').
\end{align}
Since the second term is linear in $P_{XY'}$, it suffices to focus on the first term when computing the Hessian. Note that
\begin{align}
	H(Y|Y')&=-\sum\limits_{y,y'\in\mathcal{Y}}P_{YY'}(y,y')\log P_{Y|Y'}(y|y')\nonumber\\
	&=-\sum\limits_{y,y'\in\mathcal{Y}}P_{YY'}(y,y')\log\frac{P_{YY'}(y,y')}{P_{Y'}(y')},
\end{align}
where 
\begin{align}
	P_{Y|Y'}(y|y'):=\sum\limits_{x\in\mathcal{X}}P_{X|Y'}(x|y')Q(y|x,y'),\\
	P_{YY'}(y,y'):=\sum\limits_{x\in\mathcal{X}}P_{XY'}(x,y')Q(y|x,y').
\end{align}
Therefore,
\begin{align}
	\frac{\partial H(Y|Y')}{
		\partial P_{XY'}(x,y')}=-\sum\limits_{y\in\mathcal{Y}}Q(y|x,y')\log P_{YY'}(y,y')+\log P_{Y'}(y').
\end{align}
Moreover, 
\begin{align}
	\frac{\partial^2 H(Y|Y')}{
		\partial P_{XY'}(x,y')\partial P_{XY'}(\tilde{x},\tilde{y}')}=\begin{cases}
		\frac{1}{P_{Y'}(y')}\left(1-\sum\limits_{y\in\mathcal{Y}}\frac{Q(y|x,y')Q(y|\tilde{x},y')}{P_{Y|Y'}(y|y')}\right), & y'=\tilde{y}',\\
		0, & y'\neq\tilde{y}'.
	\end{cases}
\end{align}
Hence, the Hessian $\mathbb{H}$ of $I(X;Y|Y')$ 
can be written as a block-diagonal matrix with  the $y'$-th diagonal block given by
\begin{align}
	\mathbb{H}_{y'}:=\frac{1}{P_{Y'}(y')}\left(\mathbbm{1}\mathbbm{1}^T-(\mathbb{Q}^{(c)}_{y'})^T\mathbb{D}_{y'}\mathbb{Q}^{(c)}_{y'}\right),
\end{align}
where $\mathbb{D}_{y'}:=\mathrm{diag}(\frac{1}{P_{Y|Y'}(y|y')})_{y\in\mathcal{Y}}$ for $y'\in\mathcal{Y}$. Let $\mathbbm{v}:=(\mathbbm{v}^T_0,\mathbbm{v}^T_1,\ldots,\mathbbm{v}^T_{|\mathcal{Y}|-1})$ be a column vector satisfying $\mathbbm{1}^T\mathbbm{v}=0$, where
$\mathbbm{v}_{y'}$ is a $|\mathcal{X}|$-dimensional column vector for $y'\in\mathcal{Y}$. We have
\begin{align}
	\mathbbm{v}^T\mathbb{H}\mathbbm{v}=\sum\limits_{y'\in\mathcal{Y}}	\mathbbm{v}^T_{y'}\mathbb{H}_{y'}\mathbbm{v}_{y'}.\label{eq:decompose}
\end{align} 
Let $\mathbbm{p}_{y'}:=(P_{X|Y'}(x|y'))_{x\in\mathcal{X}}^T$ and $d_{y'}:=\mathbbm{1}^T\mathbbm{v}_{y'}$ for $y'\in\mathcal{Y}$. We now show that
\begin{align}
	\mathbbm{v}^T_{y'}\mathbb{H}_{y'}\mathbbm{v}_{y'}=-\frac{1}{P_{Y'}(y')}(\mathbbm{v}_{y'}-d_{y'}\mathbbm{p}_{y'})^T(\mathbb{Q}^{(c)}_{y'})^T\mathbb{D}_{y'}\mathbb{Q}^{(c)}_{y'}(\mathbbm{v}_{y'}-d_{y'}\mathbbm{p}_{y'}).\label{eq:perspective}
\end{align}
To see this, expand
\begin{align}
	&-\frac{1}{P_{Y'}(y')}(\mathbbm{v}_{y'}-d_{y'}\mathbbm{p}_{y'})^T(\mathbb{Q}^{(c)}_{y'})^T\mathbb{D}_{y'}\mathbb{Q}^{(c)}_{y'}(\mathbbm{v}_{y'}-d_{y'}\mathbbm{p}_{y'})\nonumber\\
	&=-\frac{1}{P_{Y'}(y')}\left(\mathbbm{v}_{y'}^T(\mathbb{Q}^{(c)}_{y'})^T\mathbb{D}_{y'}\mathbb{Q}^{(c)}_{y'}\mathbbm{v}_{y'}-2d_{y'}\mathbbm{p}_{y'}^T(\mathbb{Q}^{(c)}_{y'})^T\mathbb{D}_{y'}\mathbb{Q}^{(c)}_{y'}\mathbbm{v}_{y'}+d^2_{y'}\mathbbm{p}^T_{y'}(\mathbb{Q}^{(c)}_{y'})^T\mathbb{D}_{y'}\mathbb{Q}^{(c)}_{y'}\mathbbm{p}_{y'}\right).\label{eq:expansion}
\end{align}
Since
\begin{align}
	\sum\limits_{x\in\mathcal{X}}P_{X|Y'}(x|y')\sum\limits_{y\in\mathcal{Y}}\frac{Q(y|x,y')Q(y|\tilde{x},y')}{P_{Y|Y'}(y|y')}&=\sum\limits_{y\in\mathcal{Y}}\left(\sum\limits_{x\in\mathcal{X}}P_{X|Y'}(x|y')Q(y|x,y')\right)\frac{Q(y|\tilde{x},y')}{P_{Y|Y'}(y|y')}\nonumber\\
	&=\sum\limits_{y\in\mathcal{Y}}Q(y|\tilde{x},y')\nonumber\\
	&=1,
\end{align}
it follows that
\begin{align}
	\mathbbm{p}_{y'}^T(\mathbb{Q}^{(c)}_{y'})^T\mathbb{D}_{y'}\mathbb{Q}^{(c)}_{y'}=\mathbbm{1}^T.
\end{align}
As a consequence,
\begin{align}
	&\mathbbm{p}_{y'}^T(\mathbb{Q}^{(c)}_{y'})^T\mathbb{D}_{y'}\mathbb{Q}^{(c)}_{y'}\mathbbm{v}_{y'}=d_{y'},\label{eq:subst1}\\
	&\mathbbm{p}^T_{y'}(\mathbb{Q}^{(c)}_{y'})^T\mathbb{D}_{y'}\mathbb{Q}^{(c)}_{y'}\mathbbm{p}_{y'}=1.\label{eq:subst2}
\end{align}
Substituting \eqref{eq:subst1} and \eqref{eq:subst2} into \eqref{eq:expansion} yields
\begin{align}
	&-\frac{1}{P_{Y'}(y')}(\mathbbm{v}_{y'}-d_{y'}\mathbbm{p}_{y'})^T(\mathbb{Q}^{(c)}_{y'})^T\mathbb{D}_{y'}\mathbb{Q}^{(c)}_{y'}(\mathbbm{v}_{y'}-d_{y'}\mathbbm{p}_{y'})\nonumber\\
	&=-\frac{1}{P_{Y'}(y')}\left(\mathbbm{v}_{y'}^T(\mathbb{Q}^{(c)}_{y'})^T\mathbb{D}_{y'}\mathbb{Q}^{(c)}_{y'}\mathbbm{v}_{y'}-d^2_{y'}\right)\nonumber\\
	&=-\frac{1}{P_{Y'}(y')}\left(\mathbbm{v}_{y'}^T(\mathbb{Q}^{(c)}_{y'})^T\mathbb{D}_{y'}\mathbb{Q}^{(c)}_{y'}\mathbbm{v}_{y'}-\mathbbm{v}^T_{y'}\mathbbm{1}\mathbbm{1}^T\mathbbm{v}_{y'}\right)\nonumber\\
	&=\mathbbm{v}^T_{y'}\mathbb{H}_{y'}\mathbbm{v}_{y'},
\end{align}
which proves \eqref{eq:perspective}.
It follows immediately from \eqref{eq:perspective} that
\begin{align}
	\mathbbm{v}^T_{y'}\mathbb{H}_{y'}\mathbbm{v}_{y'}\leq 0.\label{eq:negative_definite}
\end{align}
Moreover, since $\mathbb{Q}^{(c)}_{y'}$ has full rank,  equality is achieved  in \eqref{eq:negative_definite} if and only if
\begin{align} \mathbbm{v}_{y'}=d_{y'}\mathbbm{p}_{y'}.
\end{align}

Now suppose, to the contrary, that the optimization problem in \eqref{eq:feedback_capacity}--\eqref{eq:feedback_constraint} has two distinct
maximizers $P^{(0)}_{XY'}$ and $P^{(1)}_{XY'}$. Set $\mathbbm{v}_{y'}=(P^{(1)}_{XY'}(x,y')-P^{(0)}_{XY'}(x,y'))^T_{x\in\mathcal{X}}$ for $y'\in\mathcal{Y}$, and let $P^{(\alpha)}_{XY'}:=(1-\alpha)P^{(0)}_{XY'}+\alpha P^{(1)}_{XY'}$ with $\alpha\in[0,1]$. Clearly, $P^{(\alpha)}_{XY'}$ satisfies \eqref{eq:feedback_constraint}. It follows by \eqref{eq:decompose} and \eqref{eq:negative_definite}  that
\begin{align}
	\left.\mathbbm{v}^T\mathbbm{H}\mathbbm{v}\right|_{P_{XY'}=P^{(\alpha)}_{XY'}}\leq 0,\quad\alpha\in[0,1].\label{eq:lambda}
\end{align}
Moreover,  \eqref{eq:lambda} must in fact hold with equality; otherwise, there would exist some $\alpha\in(0,1)$  for which $P^{(\alpha)}_{XY'}$ yields a strictly larger value of $I(X;Y|Y')$ than both $P^{(0)}_{XY'}$ and $P^{(1)}_{XY'}$. As a consequence,
\begin{align}
	\mathbbm{v}_{y'}=d_{y'}\mathbbm{p}^{(\alpha)}_{y'},\label{eq:lambda01}
\end{align}
where $\mathbbm{p}^{(\alpha)}_{y'}:=(P^{(\alpha)}_{X|Y'}(x|y))^T_{x\in\mathcal{X}}$ for $y'\in\mathcal{Y}$. Setting $\alpha=0$ and $\alpha=1$ in \eqref{eq:lambda01} gives
$P^{(0)}_{X|Y'}=P^{(1)}_{X|Y'}$, which further implies $P^{(0)}_{Y|Y'}=P^{(1)}_{Y|Y'}$. Since $P^{(0)}_{XY'}$ and $P^{(1)}_{XY'}$ converge to $P^{(W)}_XP^{(W)}_Y$ as $\delta\rightarrow 0$, it follows that $P^{(0)}_{X|Y'}(x|y')=P^{(1)}_{X|Y'}(x|y')>0$ for all $x\in\mathcal{X}$ and $y'\in\mathcal{Y}$ when $\delta$ is sufficiently small.
Note that all channel transition matrices $\mathbb{Q}^{(c)}_0, \mathbb{Q}^{(c)}_1, \ldots, \mathbb{Q}^{(c)}_{|\mathcal{Y}|-1}$ have full rank whenever \eqref{eq:full_rank} holds. Therefore, we also have $P^{(0)}_{Y|Y'}(y|y')=P^{(1)}_{Y|Y'}(y|y')>0$ for all  $y,y'\in\mathcal{Y}$ when $\delta$ is sufficiently small. As a Markov kernel, $P^{(0)}_{Y|Y'}$ (equivalently, $P^{(1)}_{Y|Y'}$) is thus irreducible and hence admits a unique stationary distribution $P^{(k)}_{Y'}$ \cite[Theorems 5.3.3, 5.5.9, and 5.5.12]{Durrett2019}. The constraint \eqref{eq:feedback_constraint} therefore forces $P^{(0)}_{Y'}=P^{(1)}_{Y'}$, and consequently $P^{(0)}_{XY'}=P^{(1)}_{XY'}$, which contradicts the assumption that the two maximizers are distinct.

The uniqueness of the maximizer $P^*_{XY'}$ can also be established via an information-theoretic argument. Suppose, to the contrary, that the optimization problem in \eqref{eq:feedback_capacity}--\eqref{eq:feedback_constraint} admits two distinct
maximizers $P^{(0)}_{XY'}$ and $P^{(1)}_{XY'}$.  Define a new set of random variables $(\Theta,X(\Theta),Y(\Theta),Y'(\Theta))$ with the joint distribution 
\begin{align}
	P_{\Theta X(\Theta)Y(\Theta)Y'(\Theta)}(\theta,x,y,y'):=\frac{1}{2}P^{(\theta)}_{XY'}(x,y')Q(y|x,y'),
\end{align}
where $\theta\in\{0,1\}$. It can be verified that for $\theta=0,1$,
\begin{align}
	P_{X(\theta)Y(\theta)Y'(\theta)}(x,y,y')&=P^{(\theta)}_{XYY'}(x,y,y')\\
	&:=P^{(\theta)}_{XY'}(x,y')Q(y|x,y').
\end{align}
We have
\begin{align}
	I(X(\Theta);Y(\Theta)|Y'(\Theta))&=I(\Theta,X(\Theta);Y(\Theta)|Y'(\Theta))\nonumber\\
	&\geq I(X(\Theta);Y(\Theta)|Y'(\Theta),\Theta)\nonumber\\
	&=\frac{1}{2}I(X(0);Y'(0)|Y'(0))+\frac{1}{2}I(X(1);Y'(1)|Y'(1)).
\end{align}
On the other hand, since $P_{X(\Theta)Y'(\Theta)}$ also satisfy the constraint \eqref{eq:feedback_constraint}, it follows that 
\begin{align}
I(X(\Theta);Y(\Theta)|Y'(\Theta))\leq\frac{1}{2}I(X(0);Y'(0)|Y'(0))+\frac{1}{2}I(X(1);Y'(1)|Y'(1)).
\end{align}
Therefore,
\begin{align}
I(X(\Theta);Y(\Theta)|Y'(\Theta))=\frac{1}{2}I(X(0);Y'(0)|Y'(0))+\frac{1}{2}I(X(1);Y'(1)|Y'(1)),
\end{align}
which further implies
\begin{align}
I(\Theta,X(\Theta);Y(\Theta)|Y'(\Theta))= I(X(\Theta);Y(\Theta)|Y'(\Theta),\Theta).
\end{align}
As a consequence, $\Theta$ and $Y(\Theta)$ are conditionally independent given $Y'(\Theta)$, i.e., $P_{Y(0)|Y'(0)}=P_{Y(1)|Y'(1)}$, or equivalently, $P^{(0)}_{Y|Y'}=P^{(1)}_{Y|Y'}$.
Combining this with the constraint \eqref{eq:feedback_constraint}, we can
conclude that
$P^{(0)}_{X|Y'}=P^{(1)}_{X|Y'}$ (using the full rankness of $\mathbb{Q}^{(c)}_0, \mathbb{Q}^{(c)}_1, \ldots, \mathbb{Q}^{(c)}_{|\mathcal{Y}|-1}$) and $P^{(0)}_{Y'}=P^{{(1)}}_{Y'}$ (using the irreducibility of the Markov kernels $P^{(0)}_{Y|Y'}$ and $P^{(1)}_{Y|Y'}$). Hence, $P^{(0)}_{XY'}=P^{(1)}_{XY'}$, yielding a contradiction.

It is worth mentioning that, although the approach based on the analysis of the Hessian matrix is technically more involved than the information-theoretic argument, it provides additional insight into the problem. In particular, the Hessian-based analysis reveals the strong concavity of $I(X;Y|Y')$ with respect to $P_{XY'}$ over a certain subspace,  a property that is not directly apparent from the information-theoretic perspective. Such structural information is likely to be useful for establishing a strengthened version of Lemma \ref{lem:uniform_convergence} that includes an explicit estimate of the convergence rate.

%comment on the advantage of the Hessian matrix approach.

\section{Proof of Theorem \ref{thm:necessary_sufficient}}\label{app:necessary_sufficient}

The following lemma\footnote{The conditions in \eqref{eq:c_2} are partially redundant. When $n\geq 2$, the joint stationarity of $P_{X_tY_{t-1}}$ implies the marginal stationarity of
$P_{Y_{t}}$. Indeed, the joint stationary implies
 $P_{Y_{t}}=P_{Y_0}$ for $t=1,2,\ldots,n-1$, while the final case $P_{Y_{n}}=P_{Y_{n-1}}$ is a direct consequence of $P_{X_nY_{n-1}}=P_{X_{n-1}Y_{n-2}}$ and the time-invariant property of the POST channel. When $n=1$, however, the joint  stationarity is vacuous, and  $P_{Y_1}=P_{Y_0}$ must be established separately. We retain the current form of \eqref{eq:c_2} primarily for ease of application.
} plays an essential role in proving Theorem \ref{thm:necessary_sufficient} and is of interest in its own right.
\begin{lemma}\label{lem:stationary}
	For every positive integer $n$, there exists a joint input-initial-state distribution $P_{X^nY_0}$ such that the induced 
	\begin{align}
		P_{X^nY^nY_0}(x^n,y^n,y_0):=P_{X^nY_0}(x^n,y_0)\left(\prod\limits_{t=1}^nQ(y_t|x_t,y_{t-1})\right)\label{eq:c_1}
	\end{align}
	satisfies 
	\begin{align}
		\frac{1}{n}I(X^n;Y^n|Y_0)\geq C^*(Q),\label{eq:stationary_upperbound}
	\end{align}
	and
	\begin{align}
		P_{X_tY_{t-1}}=P_{X_1Y_{0}},\quad P_{Y_t}=P_{Y_0},\label{eq:c_2}
	\end{align}
	for $t=1,2,\ldots,n$.
	\end{lemma}
	\begin{IEEEproof}
	In view of \eqref{eq:nonfeedback_capacity}, given any positive integer $n$ and any $\epsilon>0$, there exists, for every sufficiently large $k$,
	a joint input-initial-state distribution $P_{X^{kn+n-1}Y_0}$ such that the induced mutual information across the POST channel $Q$ satisfies
	\begin{align}
		\frac{1}{kn+n-1}I(X^{kn+n-1};Y^{kn+n-1}|Y_0)\geq C^*(Q)-\epsilon.\label{eq:epsilon_gap}
	\end{align}
	For $i=1,2,\ldots,n$, we have
	\begin{align}
		I(X^{kn+n-1};Y^{kn+n-1}|Y_0)&=H(Y^{kn+n-1}|Y_0)-H(Y^{kn+n-1}|X^{kn+n-1},Y_0)\nonumber\\
		&=H(Y^{i-1}|Y_0)-H(Y^{i-1}|X^{kn+n-1},Y_0)\nonumber\\
		&\quad+\sum\limits_{j=1}^k\left(H(Y^{jn+i-1}_{(j-1)n+i}|Y^{(j-1)n+i-1}_0)-H(Y^{jn+i-1}_{(j-1)n+i}|X^{kn+n-1},Y^{(j-1)n+i-1}_0)\right)\nonumber\\
		&\quad+H(Y^{kn+n-1}_{kn+i}|Y^{kn+i-1}_0)-H(Y^{kn+n-1}_{kn+i}|X^{kn+n-1},Y^{kn+i-1}_0)\nonumber\\
		&=H(Y^{i-1}|Y_0)-H(Y^{i-1}|X^{i-1},Y_0)\nonumber\\
		&\quad+\sum\limits_{j=1}^k\left(H(Y^{jn+i-1}_{(j-1)n+i}|Y^{(j-1)n+i-1}_0)-H(Y^{jn+i-1}_{(j-1)n+i}|X^{jn+i-1}_{(j-1)n+i},Y_{(j-1)n+i-1})\right)\nonumber\\
		&\quad+H(Y^{kn+n-1}_{kn+i}|Y^{kn+i-1}_0)-H(Y^{kn+n-1}_{kn+i}|X^{kn+n-1}_{kn+i},Y_{kn+i-1})\nonumber\\
		& \leq H(Y^{i-1}|Y_0)-H(Y^{i-1}|X^{i-1},Y_0)\nonumber\\
		&\quad+\sum\limits_{j=1}^k\left(H(Y^{jn+i-1}_{(j-1)n+i}|Y_{(j-1)n+i-1})-H(Y^{jn+i-1}_{(j-1)n+i}|X^{jn+i-1}_{(j-1)n+i},Y_{(j-1)n+i-1})\right)\nonumber\\
		&\quad+H(Y^{kn+n-1}_{kn+i}|Y_{kn+i-1})-H(Y^{kn+n-1}_{kn+i}|X^{kn+n-1}_{kn+i},Y_{kn+i-1})\nonumber\\
		&=H(X^{i-1};Y^{i-1}|Y_0)+\sum\limits_{j=1}^kI(X^{jn+i-1}_{(j-1)n+i};Y^{jn+i-1}_{(j-1)n+i}|Y_{(j-1)n+i-1})\nonumber\\
		&\quad+I(X^{kn+n-1}_{kn+i};Y^{kn+n-1}_{kn+i}|Y_{kn+i-1})\nonumber\\
		&\leq(n-1)\log|\mathcal{X}|+\sum\limits_{j=1}^kI(X^{jn+i-1}_{(j-1)n+i};Y^{jn+i-1}_{(j-1)n+i}|Y_{(j-1)n+i-1}).
	\end{align}
	Therefore,
	\begin{align}
		nI(X^{kn+n-1};Y^{kn+n-1}|Y_0)\leq n(n-1)\log|\mathcal{X}|+\sum\limits_{i=1}^{n}\sum\limits_{j=1}^kI(X^{jn+i-1}_{(j-1)n+i};Y^{jn+i-1}_{(j-1)n+i}|Y_{(j-1)n+i-1}),
	\end{align}
	which  implies
	\begin{align}
		&\frac{1}{kn+n-1}I(X^{kn+n-1};Y^{kn+n-1}|Y_0)\nonumber\\
		&\leq\frac{n-1}{kn+n-1}\log|\mathcal{X}|+\frac{1}{n(kn+n-1)}\sum\limits_{i=1}^{n}\sum\limits_{j=1}^kI(X^{jn+i-1}_{(j-1)n+i};Y^{jn+i-1}_{(j-1)n+i}|Y_{(j-1)n+i-1}).\label{eq:besubst}
	\end{align}
	Let $\Theta$ be uniformly distributed over $\{0,1,\ldots,kn-1\}$ and independent of $(X^{kn+n-1},Y^{kn+n-1}_0)$. Define $X_i(\Theta):=X_{\Theta+i}$ for $i=1,2,\ldots,n$, and $Y_i(\Theta):=Y_{\Theta+i}$ for $i=0,1,\ldots,n$. Then
	\begin{align}
		\sum\limits_{i=1}^{n}\sum\limits_{j=1}^kI(X^{jn+i-1}_{(j-1)n+i};Y^{jn+i-1}_{(j-1)n+i}|Y_{(j-1)n+i-1})&=\sum\limits_{i=1}^{n}\sum\limits_{j=1}^kI(X^n(\Theta);Y^n(\Theta)|Y_0(\Theta),\Theta=(j-1)n+i-1)\nonumber\\
		&=knI(X^n(\Theta);Y^n(\Theta)|Y_0(\Theta),\Theta)\nonumber\\
		&\leq knI(\Theta,X^n(\Theta);Y^n(\Theta)|Y_0(\Theta))\nonumber\\
		&=knI(X^n(\Theta);Y^n(\Theta)|Y_0(\Theta)).\label{eq:subst}
	\end{align}
	Substituting \eqref{eq:subst} into \eqref{eq:besubst} gives
	\begin{align}
		\frac{1}{kn+n-1}I(X^{kn+n-1};Y^{kn+n-1}|Y_0)\leq\frac{n-1}{kn+n-1}\log|\mathcal{X}|+\frac{k}{kn+n-1}I(X^n(\Theta);Y^n(\Theta)|Y_0(\Theta)).\label{eq:k_limit}
	\end{align}
	Note that the joint distribution of $(X^n(\Theta),Y^n(\Theta),Y_0(\Theta))$ factors as
	\begin{align}
		P_{X^n(\Theta)Y^n(\Theta)Y_0(\Theta)}(x^n,y^n,y_0)=P_{X^n(\Theta)Y_0(\Theta)}(x^n,y_0)\left(\prod\limits_{t=1}^nQ(y_t|x_t,y_{t-1})\right),\label{eq:Pn_factor}
	\end{align}
	where
	\begin{align}
		P_{X^n(\Theta)Y_0(\Theta)}(x^n,y_0)=\frac{1}{nk}\sum\limits_{i=1}^n\sum\limits_{j=1}^kP_{X_{(j-1)n+i}^{jn+i-1}Y_{(j-1)n+i-1}}(x^n,y_0).
	\end{align}
	Moreover, for $t=1,2,\ldots,n$,
	\begin{align}
		&d_{\mathrm{TV}}(P_{X_t(\Theta)Y_{t-1}(\Theta)},P_{X_{1}(\Theta)Y_{0}(\Theta)})\nonumber\\
		&=\frac{1}{2}\sum\limits_{x\in\mathcal{X},y\in\mathcal{Y}}\left|P_{X_t(\Theta)Y_{t-1}(\Theta)}(x,y)-P_{X_{1}(\Theta)Y_{0}(\Theta)}(x,y)\right|\nonumber\\
		&=\frac{1}{2nk}\sum\limits_{x\in\mathcal{X},y\in\mathcal{Y}}\left|\sum\limits_{i=1}^n\sum\limits_{j=1}^kP_{X_{(j-1)n+i+t-1}Y_{(j-1)n+i+t-2}}(x,y)-\sum\limits_{i=1}^n\sum\limits_{j=1}^kP_{X_{(j-1)n+i}Y_{(j-1)n+i-1}}(x,y)\right|\nonumber\\
		&=\frac{1}{2nk}\sum\limits_{x\in\mathcal{X},y\in\mathcal{Y}}\left|\sum\limits_{\ell=kn+1}^{kn+t-1}P_{X_{\ell}Y_{\ell-1}}(x,y)-\sum\limits_{\ell=1}^{t-1}P_{X_{\ell}Y_{\ell-1}}(x,y)\right|\nonumber\\
		&\leq\frac{t-1}{nk}|\mathcal{X}||\mathcal{Y}|\nonumber\\
		&\leq\frac{n-1}{nk}|\mathcal{X}||\mathcal{Y}|\label{eq:TV}
	\end{align}
	and
	\begin{align}
	&d_{\mathrm{TV}}(P_{Y_{t}(\Theta)},P_{Y_{0}(\Theta)})\nonumber\\
	&=\frac{1}{2}\sum\limits_{y\in\mathcal{Y}}\left|P_{Y_{t}(\Theta)}(y)-P_{Y_{0}(\Theta)}(y)\right|\nonumber\\
	&=\frac{1}{2nk}\sum\limits_{y\in\mathcal{Y}}\left|\sum\limits_{i=1}^n\sum\limits_{j=1}^kP_{Y_{(j-1)n+i+t-1}}(y)-\sum\limits_{i=1}^n\sum\limits_{j=1}^kP_{Y_{(j-1)n+i-1}}(y)\right|\nonumber\\
	&=\frac{1}{2nk}\sum\limits_{y\in\mathcal{Y}}\left|\sum\limits_{\ell=kn}^{kn+t-1}P_{Y_{\ell}}(y)-\sum\limits_{\ell=0}^{t-1}P_{Y_{\ell}}(y)\right|\nonumber\\
	&\leq\frac{t}{nk}|\mathcal{Y}|\nonumber\\
	&\leq\frac{1}{k}|\mathcal{Y}|,\label{eq:TV2}
	\end{align}
	where $d_{\mathrm{TV}}(\cdot,\cdot)$ denotes the total variation distance.
	In view of \eqref{eq:epsilon_gap}, \eqref{eq:k_limit}, \eqref{eq:Pn_factor},  \eqref{eq:TV}, and \eqref{eq:TV2}, sending $k\rightarrow\infty$ then $\epsilon\rightarrow0$ proves the desired result.
		\end{IEEEproof}

With Lemma \ref{lem:stationary} in hand, we proceed to the proof of Theorem \ref{thm:necessary_sufficient}. First, consider the `if' direction.
 Given a joint input-initial-state distribution $P_{X^nY_0}$ satisfying \eqref{eq:first_condition} and \eqref{eq:second_condition}, 
we have
\begin{align}
	\frac{1}{n}I(X^n;Y^n|Y_0)&=\frac{1}{n}H(Y^n|Y_0)-\frac{1}{n}H(Y^n|X^n,Y_0)\nonumber\\
	&=\frac{1}{n}\sum\limits_{t=1}^nH(Y_t|Y^{t-1}_0)-\frac{1}{n}\sum\limits_{t=1}^nH(Y_t|X^n,Y^{t-1}_0)\nonumber\\
	&=\frac{1}{n}\sum\limits_{t=1}^nH(Y_t|Y^{t-1}_0)-\frac{1}{n}\sum\limits_{t=1}^nH(Y_t|X_t,Y_{t-1})\nonumber\\
	&\stackrel{(a)}{=}\frac{1}{n}\sum\limits_{t=1}^nH(Y_t|Y_{t-1})-\frac{1}{n}\sum\limits_{t=1}^nH(Y_t|X_t,Y_{t-1})\nonumber\\
	&=\frac{1}{n}\sum\limits_{t=1}^nI(X_t;Y_t|Y_{t-1})\nonumber\\
	&\stackrel{(b)}{=}C^*_f(Q),\label{eq:hold_for_all_n}
\end{align}
where ($a$) and ($b$) are due to \eqref{eq:first_condition} and \eqref{eq:second_condition}, respectively. Since an input-initial-state distribution satisfying \eqref{eq:hold_for_all_n} exists for every positive integer
 $n$, it follows that 
\begin{align}
	C^*(Q)\geq C^*_f(Q).
\end{align}
On the other hand, invoking Lemma \ref{lem:stationary} with $n=1$ yields
\begin{align}
	C^*(Q)\leq C^*_f(Q).
\end{align}
This completes the proof of the `if' direction.

	%It suffices to prove the ``only if" part. When $\delta$ is small enough, the POST channel $Q$ is both indecomposable and connected; therefore, the characterizations of the non-feedback capacity in \eqref{eq:nonfeedback_capacity} and the feedback capacity in \eqref{eq:feedback_capacity} apply.

Next, consider `only if' direction. 
By Lemma \ref{lem:stationary}, for every positive integer $n$,
 there exists a joint input-initial-state distribution $P_{X^nY_0}$ such that the induced 
	$P_{X^nY^nY_0}$
satisfies \eqref{eq:stationary_upperbound}
and \eqref{eq:c_2}.
If $C^*(Q)=C^*_f(Q)$, then 
\begin{align}
	\frac{1}{n}I(X^n;Y^n|Y_0)\geq  C^*_f(Q).\label{eq:comb_1}
\end{align}
Note that 
\begin{align}
\frac{1}{n}I(X^n;Y^n|Y_0)&=\frac{1}{n}H(Y^n|Y_0)-\frac{1}{n}H(Y^n|X^n,Y_0)\nonumber\\
&=\frac{1}{n}\sum\limits_{t=1}^nH(Y_t|Y^{t-1}_0)-\frac{1}{n}\sum\limits_{t=1}^nH(Y_t|X^n,Y^{t-1}_0)\nonumber\\
&=\frac{1}{n}\sum\limits_{t=1}^nH(Y_t|Y^{t-1}_0)-\frac{1}{n}\sum\limits_{t=1}^nH(Y_t|X_t,Y_{t-1})\nonumber\\
&\leq\frac{1}{n}\sum\limits_{t=1}^nH(Y_t|Y_{t-1})-\frac{1}{n}\sum\limits_{t=1}^nH(Y_t|X_t,Y_{t-1})\nonumber\\
&=\frac{1}{n}\sum\limits_{t=1}^nI(X_t;Y_t|Y_{t-1})\nonumber\\
&\stackrel{(c)}{=}I(X_1;Y_1|Y_0),\label{eq:comb_2}
\end{align}
where ($c$) is due to \eqref{eq:c_2}.
Combining \eqref{eq:comb_1} and \eqref{eq:comb_2} yields
\begin{align}
	I(X_1;Y_1|Y_0)\geq C^*_f(Q).
\end{align}
On the other hand, it follows from \eqref{eq:c_1} and \eqref{eq:c_2} that
\begin{align}
	P_{X_1Y_1Y_0}(x,y,y')=P_{X_1Y_0}(x,y')Q(y|x,y')
\end{align}
and $P_{X_1Y_0}$
satisfies the constraint \eqref{eq:feedback_constraint}. As a consequence, 
\begin{align}
I(X_1;Y_1|Y_0)\leq C^*_f(Q).
\end{align}
Therefore, 
\begin{align}
I(X_1;Y_1|Y_0)=C^*_f(Q),
\end{align}
i.e., $P_{X_1Y_0}$ is  a maximizer of the optimization problem in \eqref{eq:feedback_capacity}--\eqref{eq:feedback_constraint}.
Hence, the only inequality in  \eqref{eq:comb_2} must in fact hold with equality, which forces
\begin{align}
	P_{Y^nY_0}(y^n,y_0)&=P_{Y_0}(y_0)\prod\limits_{t=1}^nP_{Y_t|Y_{t-1}}(y_t|y_{t-1})\nonumber\\
	&\stackrel{(d)}{=}P_{Y_0}(y_0)\prod\limits_{t=1}^nP_{Y_1|Y_0}(y_t|y_{t=1}),\label{eq:forced_Markov}
\end{align}
where ($d$) is due to \eqref{eq:c_2}.
%Moreover, when $\delta$ is sufficiently small,  Lemma \ref{lem:refined_version} implies that $\hat{P}_{X_1Y_0}=P^*_{X_1Y_0}$, which in turn implies
%5\begin{align} \hat{P}_{Y_1|Y_0}=\hat{P}^*_{Y|Y'}.\label{eq:c_3}
%	\end{align}
Combing \eqref{eq:c_2} and \eqref{eq:forced_Markov} proves that $P_{X^nY_0}$ has the desired properties.

%\section{Proof of Theorem \ref{thm:conditional}}\label{app:conditional}

%\section*{Acknowledgment}


\ifCLASSOPTIONcaptionsoff
  \newpage
\fi


\begin{thebibliography}{1}
	
	\bibitem{Shannon56}
	C. Shannon, ``The zero error capacity of a noisy channel," {\em  IRE Trans. Inf. Theory}, vol.~2, no.~3, pp.~8--19, Sep. 1956.
	
		\bibitem{CT91}
	T. M. Cover and J. A. Thomas, {\em Elements of Information Theory}. New York, NY, USA: Wiley, 1991.
	

	
	\bibitem{Alajaji95}
F. Alajaji, ``Feedback does not increase the capacity of discrete channels
with additive noise," {\em IEEE Trans. Inf. Theory}, vol.~41, no.~2,
pp.~546--549, Mar. 1995.


	
	
	\bibitem{PAW14}
	H. H. Permuter, H. Asnani and T. Weissman,, ``Capacity of a POST channel with and
	without feedback," {\em IEEE Trans. Inf. Theory}, vol.~60, no.~10, pp.~6041--6057, Oct. 2014.
	
	
	\bibitem{BBT58}
	D. Blackwell, L. Breiman, and A. J. Thomasian, ``Proof of Shannon’s
	transmission theorem for finite-state indecomposable channels," {\em Ann.
		Math. Statist.}, vol.~29, pp.~1209--1220, Dec. 1958.
		
			\bibitem{Gallager68}
		R. G. Gallager, {\em Information Theory and Reliable Communication},
		New York, NY, USA: Wiley, 1968.
	
	
	\bibitem{SSP22}
	E. Shemuel, O. Sabag and H. H. Permuter, ``The feedback capacity of noisy output Is the sTate (NOST) channels," {\em IEEE Trans. Inf. Theory}, vol.~68, no.~8, pp.~5044-5059, Aug. 2022.
	
	
	\bibitem{CB05}
	J. Chen and T. Berger, ``The capacity of finite-state Markov channels with feedback," {\em IEEE Trans. Inf. Theory}, vol.~51, no.~3, pp.~780--798, Mar. 2005.
	
	\bibitem{Hajnal58}	
	J. Hajnal, ``Weak ergodicity in non-homogeneous Markov chains,"
	{\em Proc. Combridge Philos. Soc.}, vol.~54, pp.~233--246, 1958.
	
	\bibitem{CPW10}
	J. Chen, H. Permuter and T. Weissman, ``Tighter bounds on the capacity of finite-state channels via Markov set-chains,"  {\em IEEE Trans. Inf. Theory}, vol.~56, no.~8, pp.~3660--3691, Aug. 2010.
	

	


	
%\bibitem{Seneta79}
%E. Seneta, ``Coefficients of ergodicity: Structure and applications," {\em Adv.
%Appl. Probab.}, vol.~11, pp.~576--590, Sep. 1979.



\bibitem{Durrett2019}
R. Durrett, {\em Probability: Theory and Examples}, 5th ed. Cambridge, U.K.:
Cambridge Univ. Press, 2019.
	
	
	
	
	
	
\end{thebibliography}
\end{document}